\documentclass[11pt,letterpaper]{article}
\usepackage{mehdis}

\newif\ifnotes\notestrue
\ifnotes
\usepackage{color}
\definecolor{mygrey}{gray}{0.50}
\newcommand{\notename}[2]{{\textcolor{mygrey}{\footnotesize{\bf (#1:} {#2}{\bf ) }}}}

\newcommand{\pnote}[1]{{\endnote{#1}}}

\else

\newcommand{\notename}[2]{{}}

\newcommand{\pnote}[1]{}

\fi

\begin{document}
\title{\LARGE{From communication complexity to an entanglement spread area law in the ground state of gapped local Hamiltonians}
}
\author{Anurag Anshu\thanks{Institute for Quantum Computing and Department of Combinatorics and Optimization, University of Waterloo, Canada and Perimeter Institute for Theoretical Physics, Canada. \href{mailto:aanshu@uwaterloo.ca}{aanshu@uwaterloo.ca}} \and Aram W. Harrow\thanks{Center for Theoretical Physics, MIT. \href{mailto:aram@mit.edu}{aram@mit.edu}}
 \and 
 Mehdi Soleimanifar\thanks{Center for Theoretical Physics, MIT. \href{mailto:mehdis@mit.edu}{mehdis@mit.edu}} }

\date{\today}

\maketitle
 \begin{abstract}
 In this work, we make a connection between two seemingly different problems. The first problem involves characterizing the properties of entanglement in the ground state of gapped local Hamiltonians, which is a central topic in quantum many-body physics. The second problem is on the quantum communication complexity of testing bipartite states with EPR assistance, a well-known question in quantum information theory. We construct a communication protocol for testing (or measuring) the ground state and use its communication complexity to reveal a new structural property for the ground state entanglement. This property, known as the entanglement spread, roughly measures the log of the ratio between the largest and the smallest Schmidt coefficients across a bipartite cut in the ground state. Our main result shows that gapped ground states possess limited entanglement spread across any cut, exhibiting an ``area law'' behavior. 
 
 Our result quite generally applies to any interaction graph with an improved bound for the special case of lattices. This entanglement spread area law includes interaction graphs constructed in \cite{Aharonov2014CounterExampleToAreaLaw} that violate a generalized area law for the entanglement \emph{entropy}. Our construction also provides evidence for a conjecture in physics by Li and Haldane on the entanglement spectrum of lattice Hamiltonians \cite{LiHaldane:2008}.  
 
 On the technical side, we use recent advances in Hamiltonian simulation algorithms along with the quantum phase estimation to give a new construction for an approximate ground space projector (AGSP) over arbitrary interaction graphs, which might be of independent interest.
 \end{abstract}
%%%%%%%%%%%%%%%%%%%%%%%%%%%%%%%%%%%%%%%%%%%%%%%%%%%%%%%%%%%%%%%%%%%%%%%%%%%%%%%%%%%%%%%%%%%%%%%%%%%
\newpage
\section{Introduction}
\subsection{Background on area law and entanglement spectra}
The ground states of local Hamiltonians are examples of quantum many-body states with central significance in condensed matter physics and quantum chemistry. A crucial distinction between these states and their classical counterparts -- the satisfying assignments in constraint satisfaction problems -- is the presence of multipartite entanglement. This leads to novel phenomena in these systems such as exotic phases of matter, but also complicates the theoretical and numerical study of their properties.

There is a successful line of research that applies the tools developed in quantum information theory and computer science to study various features of entanglement in the ground states. An important problem that has been the focus of many such studies is proving a conjecture known as the ``area law'' for the entanglement entropy in the ground state of gapped local Hamiltonians. We can more precisely state this by considering the interaction (hyper)graph where each vertex represents a qudit and the edges correspond to the interaction terms in the Hamiltonian (see \fig{1}). Suppose we fix a partition of the qudits into two parts $A$ and $B$. We denote the ground state by $\ket{\Omega}_{AB}$. In general, the qudits in part $A$ will be entangled with those in part $B$. The area law asserts that the amount of entanglement -- measured by the entropy of either of the reduced states $\Omega_A$ or $\Omega_B$ -- is at most proportional to the number of interaction terms that cross the cut $\partial A$. This behavior is drastically different from the generic situation where the entanglement across the cut $\partial A$ scales with the size of the smaller partition $|A|$ rather than $|\partial A|$. Thus, loosely speaking, the area law implies that the ground state entanglement is local and limited to the boundary. This conjecture has been rigorously proven when the interaction graph is a 1D chain \cite{hastings2007area_law,one-dimensional_area_law,area_law_subexponential_algo} and there has been recent progress on trees \cite{abrahamsen2019polynomial} and 2D lattices \cite{anshu2019AreaLaw2D}. 

A generalization of this conjecture asks if the area law holds for arbitrary interaction graphs beyond lattices. It turns out that this generalized conjecture is false. Using quantum expanders, an interaction graph is constructed in \cite{Aharonov2014CounterExampleToAreaLaw} which admits a partition into two parts $A$ and $B$ such that the size of the cut is $|\partial A|=1$, but the amount of entanglement across the cut is proportional to $|A|$. 

Thus far, these results study the ground state entanglement in terms of the \emph{entropy} of the reduced state $\Omega_A$ on partition $A$. One can go beyond this and consider other features of the eigenvalues of the reduced state $\Omega_A$ -- the log of which is known as the \emph{entanglement spectrum} -- besides its entropy. Most notably, in \cite{LiHaldane:2008}, Li and Haldane conjectured that the entanglement spectrum of 2D gapped ground states in a region $A$ resembles the spectrum of the Gibbs state of a local Hamiltonian acting only on the boundary $\partial A$. This Hamiltonian is often called the \emph{modular Hamiltonian} $\modH$ and the Gibbs state is the state proportional to $e^{-\modH}$. This conjecture, which is stronger that the area law discussed before, has been extensively studied both numerically and theoretically in several works \cite{Schuch_boundary_hamiltonian_2013, Cirac_boundary_hamiltonian_2011, KatoBrandao_edge_state2019}.

Inspired by these results, we prove a new structural property for the entanglement of gapped ground states. The key to our findings is a connection to the field of quantum communication complexity. A basic question there is when two parties want to \emph{test} whether they share a specific entangled state $\ket\psi$ by exchanging as few messages as possible.  In other words, they want to perform the measurement $\{\ketbra{\psi}{\psi},\iden-\ketbra{\psi}{\psi}\}$. Building on \cite{state_conversion_coudron}, we resolve the communication complexity of this problem and relate it to the details of the entanglement in $\ket\psi$. To apply this to gapped Hamiltonians, we choose $\ket\psi$ to be be the ground state $\ket{\Omega}$ of a local Hamiltonian. Then designing a testing protocol for $\ket{\Omega}$ tells us about the nature of the ground state entanglement. We devise such a measurement protocol tailored for the ground state of \emph{gapped} local Hamiltonians by combining recent Hamiltonian simulation techniques with the quantum phase estimation algorithm, which might be of independent interest.

The property that we study is known as the \emph{entanglement spread}, which roughly measures the log of the ratio between the largest and the smallest eigenvalue of the reduced state $\Omega_A$, giving an estimate of how spread out their distribution is (see \fig{1}). Our results quite generally apply to \emph{any} interaction graph, with some improved statement for the special case of lattices. We show that as long as the Hamiltonian is gapped, its ground state possesses limited entanglement spread on general interaction graphs, exhibiting an ``area law'' behavior. On lattices, we prove a sub-area scaling for this quantity. We use these results to give formal evidence for the aforementioned conjecture by Li and Haldane about modular Hamiltonians. We also show that both states that satisfy the entropy area law and those in the counter-example construction in \cite{Aharonov2014CounterExampleToAreaLaw} fit into our framework. In the next sections, we provide a more detailed overview of our setup and results. 

 \begin{figure}[t!]
 \centering
 	\includegraphics[width=.75
 \textwidth]{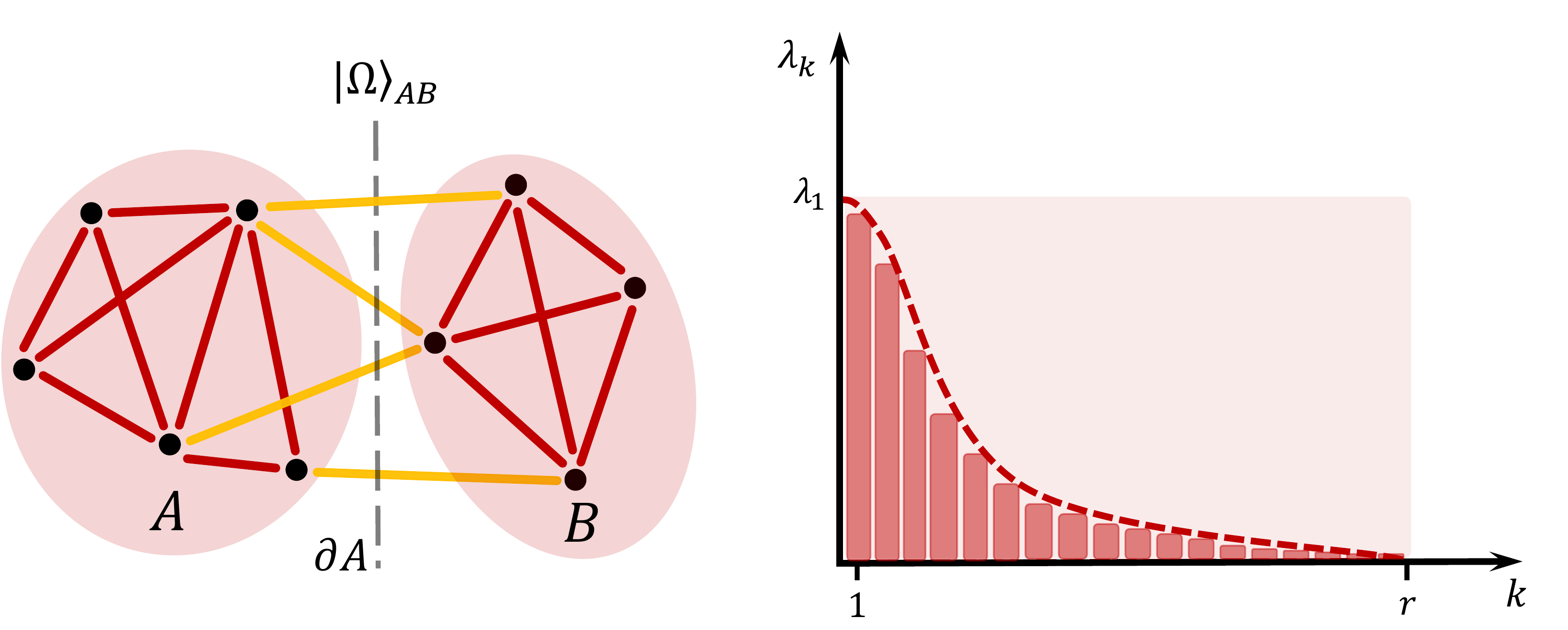}
 	\caption{(Left) Consider a partition of the ground state into two parts. What does the entanglement across a cut look like? (Right) The profile of the eigenvalues (Schmidt coefficients) of the reduced state in region $A$. The entanglement spread across the cut $\partial A$ is roughly the log of the ratio between the largest and smallest Schmidt coefficient. We show that the entanglement spread scales as $ O(|\partial A|)$, while on a lattice, we improve this to $O(\sqrt{|\partial A|})$. 
 	}\label{fig:1}
\end{figure}

\subsection{Entanglement spread and communication complexity of non-local measurements}
Consider a bipartite state $\ket{\psi}\in \cH_{A}\ot \cH_B$, where Alice and Bob own registers, $\cH_A$ and $\cH_B$ respectively. Suppose the parties engage in a communication protocol whose goal is to test if they share the state $\ket{\psi}$. That is, they would like to implement the reflection operator $\Refl(\psi)=2\ketbra{\psi}{\psi}-\iden$ or similarly, perform the two-outcome measurement  $\{\ketbra{\psi}{\psi},\iden-\ketbra{\psi}{\psi}\}$. In our setup, these tasks are locally interchangeable. Namely, the ability to perform controlled reflections will give us the ability to do measurements and vice versa \footnote{To see this, just locally initialize a qubit in state $\ket{+}$, perform a controlled reflection, and then locally measure in the Hadamard basis. The ability to do a coherent measurement also gives the power of reflection: We can add a $-1$ phase if the outcome of the measurement is $\iden - \ketbra{\psi}{\psi}$.}. Since in general, the state $\ket{\psi}$ is an entangled state, Alice and Bob need to exchange qubits to perform this operation. For instance, Alice can send her register $\cH_A$ to Bob who then performs the joint operation on $\cH_A \ot \cH_B$ and sends back Alice's register. As we will see later, they can often do much better. The communication complexity (or cost) of such a protocol $\CC_{\D}(\psi)$ is defined as the \emph{minimum} number of qubits that the parties need to exchange to perform this task with error at most $\D$. Here, we are interested in the case where $\ket{\psi}$ is the ground state of a gapped local Hamiltonian. In other words, we ask

\begin{quote}
    \begin{center}
        \emph{What is the communication cost $\mathit{C_{\D}(\Omega)}$ of approximately measuring or reflecting about the ground state of a gapped local Hamiltonian which is shared between Alice and Bob? }
    \end{center}
\end{quote}

Before specializing to the ground state, it is insightful to consider a few general instances. 
\begin{itemize}
\item[1)] $\ket{\psi}=\ket{0}^{\ot n}\ket{0}^{\ot n}$: This is equivalent to a CZ gate, or equivalently a CNOT gate, which has communication cost 1.
\item[2)] $\ket{\psi}=\frac{1}{\sqrt{p}}\sum_{j=1}^p \ket{j}\ket{j}$: Initially, one might think that reflecting about the maximally entangled state $\ket{\Phi_p}=\frac{1}{\sqrt{p}}\sum_{j=1}^p \ket{j}\ket{j}$ requires exchanging a large number of qubits, but it turns out that by using quantum expanders \cite{Aharonov2014CounterExampleToAreaLaw}, one can perform such a reflection up to error $\D$ by exchanging $\CC_{\D}(\psi)=O(\log(1/\D))$ qubits, which is independent of the dimension $p$. 
\item[3)] $\ket{\psi}=\frac{1}{\sqrt{2}}(\ket{00}^{\ot n}+\ket{\Phi_2}^{\ot n})$: This is a superposition of the last two cases with $\ket{\Phi_2}=\frac{1}{\sqrt{2}}(\ket{00}+\ket{11})$ being the EPR pair. We claim that $\CC_{\D}(\psi)=\Theta(n)$. This can be verified by noticing that $\norm{\Refl(\psi)\ket{00}^{\ot n}-\ket{\Phi_2}^{\ot n}}\leq 2^{-(n-1)/2}$, but it is well-known that creating $n$ EPR pairs $\ket{\Phi_2}^{\ot n}$ from product states requires $\Omega(n)$ qubits of communication; see also \cite{HL07}.
\end{itemize}

In general, by applying local unitaries, a bipartite state $\ket{\psi}$ can be always transformed to a standard form $\ket{\psi}=\sum_{j=1}^r \sqrt{\l}_j \ket{j}\ket{j}$ known as the Schmidt decomposition. Thus, we expect $\CC_{\D}(\psi)$ to depend only on the Schmidt coefficients $\l_j$. We assume these coefficients are arranged in the descending order $\l_1\geq \l_2\geq \dots \geq \l_r$, where $r$ is called the Schmidt rank of $\ket{\psi}$. A closer look at the above examples reveals a pattern. In the first two instances, where $\CC_{\D}(\psi)$ is small, the Schmidt coefficients of $\ket{\psi}$ are all equal (which we refer to as the \emph{concentrated} case). In the third example, which has high communication cost, the Schmidt coefficients are \emph{spread out} between two different values $\frac{1}{2}(1+\frac{1}{\sqrt{2^{n}}})^2 \approx 1/2$ and $\frac{1}{2^{n+1}}$ with almost equal weights. 

This motivates a more general lower bound on the communication complexity in terms of the \emph{entanglement spread} of the state $\ket{\psi}$, which is a measure of how spread out the Schmidt coefficients are across a cut \cite{cost_entanglement_transformations_Hayden}.
In its simplest form, the entanglement spread, denoted by $\ES({\psi})$, is defined by 
\begin{align}\label{eq:e4}
 \ES({\psi})= \log(r \l_1)=S_{\max}(\psi_A)-S_{\min}(\psi_A),
\end{align}
where $\psi_A=\Tr_B\ketbra{\psi}{\psi}$ and
\begin{align}
    S_{\max}(\psi_A)=\log r, \quad S_{\min}(\psi_A)=-\log \l_1 \label{eq:e6}
\end{align} 
are the max- and min-entropies given in terms of the Schmidt rank $r$ and the maximum Schmidt coefficient $\l_1$ (see \fig{1}). Indeed, one can verify that in the first two examples, $\ES({\psi})=0$ while in the third case, $\ES({\psi})=\Theta(n)$. Thus, $\ES(\psi)$ distinguishes between the concentrated versus spread out cases in the examples above. It might also be helpful to consider the entanglement spread roughly as $\log(\l_1/\l_r)$, the log of the ratio between the largest and smallest Schmidt coefficients. This makes the connection between the spread of the spectrum of $\psi$ and $\ES(\psi)$ clearer.

In our results, we need a more robust version of the entanglement spread \eqref{eq:e4} that applies to protocols which only \emph{approximately} implement the two-outcome ground state measurement, i.e. when the error $\D>0$. We denote this version by $\ES_{\d}(\psi)$ and following \cite{cost_entanglement_transformations_Hayden}, we define:
\begin{definition}[Entanglement spread]
Let $\d\in[0,1)$. The ($\d$-smooth) entanglement spread of a bipartite state $\ket{\psi}\in \cH_{AB}$ is defined by
 \begin{align}
    \ES_{\d}(\psi)=S_{\max}^{\d}(\psi_A)-S_{\min}^\d(\psi_A),\nn
\end{align}
where $S_{\min}^{\d}(\psi_A)$ and $S_{\max}^{\d}(\psi_A)$ are the smooth min- and max- entropies defined similar to \eqref{eq:e6} after removing up to a mass $\d$ from the Schmidt distribution of $\ket{\psi}$ (see \defref{smooth ent} in the body for details).  
\end{definition}

But besides the above examples, why does the entanglement spread provide a lower bound on the communication cost of measuring (testing) a \emph{general} state $\ket{\psi}$? In the exact case $\D=0$, this can be seen by observing that for each qubit exchanged between Alice and Bob, the rank $r$ and the largest Schmidt coefficient $\l_1$ change at most by a factor of $2$ and hence, after $c$ rounds of communication, $\log(r\l_1)$ is at most $2c$ \cite{cost_entanglement_transformations_Hayden}. This shows that (modulo a constant) $\ES(\psi)$ provides a lower bound on the exact communication complexity. 

In the approximate regime $\D>0$, similar lower bounds in terms of $\ES_{\d}(\psi)$ have been proved before \cite{cost_entanglement_transformations_Hayden, state_conversion_coudron}. We derive an analogous lower bound tailored for when the state of interest is the ground state of a gapped local Hamiltonian. Before stating our bound, we explain the setup in more detail.

A very useful property of the entanglement spread is that it remains unchanged when the state is supplemented by arbitrary numbers of EPR pairs. That is, $\ES_{\d}({\psi}\ot {\Phi_p})=\ES_{\d}({\psi})$ for any maximally entangled state of arbitrary size $p$. One way of seeing this is that adding $\ket{\Phi_p}$ multiplies $r$ by $2^p$ and divides $\l_1$ by $2^p$, leaving the entanglement spread unchanged. One interesting implication of this equality is that the lower bound on the communication complexity in terms of $\ES_{\d}(\psi)$ continues to hold even in protocols where Alice and Bob share an arbitrary number of EPR pairs during their communication. As we will see later, this improves and simplifies our analysis. We denote the shared EPR pairs collectively by $\ket{\Phi}$ and call such protocols \emph{EPR-assisted}. 

\section{Our results}
\subsection{Lower bound on communication complexity from entanglement spread}
As described in the previous section, the goal of Alice and Bob is to implement an operator $K$ that acts jointly on an input state and the shared EPR pairs $\ket{\Phi}$ and approximately projects the input state onto the ground state $\ket{\Omega}$ while leaving  $\ket{\Phi}$ untouched. More precisely, we define:
\begin{definition}[EPR-assisted AGSP]\label{def:ERP AGSP informal}
An EPR-assisted Approximate Ground State Projector (EPR-assisted AGSP for short) associated with the ground state $\ket{\Omega}$ of a local Hamiltonian is an operator $K$ that for some error $\D\in (0,1)$ satisfies
\begin{align}
\norm{(K-\iden \otimes \ketbra{\Omega}{\Omega})\ \ket{\Phi} \ket{\psi}} \leq \Delta \quad \text{for all}\quad  \ket{\psi}\in \cH_A \ot \cH_B .\nn
\end{align}
\end{definition}
We often also equivalently write $K(\ket{\Phi}\otimes \iden)\approx_\Delta \ket{\Phi}\otimes\ketbra{\Omega}{\Omega}$, where the notation $A\approx_\Delta B$ means that $\norm{A-B}\leq \Delta$.  In our first result, we give a lower bound on the communication complexity of implementing this operator in terms of the entanglement spread of the ground state. 

\begin{thm}[Lower bound on the complexity of EPR-assisted AGSP]\label{thm:upperbound on spread infromal}
 Let $\ket{\Omega}\in \cH_A \ot \cH_B$ be the ground state of a local Hamiltonian shared between Alice and Bob. For any error $\D\in (0,1)$, the communication complexity of implementing the two-outcome measurement $\{K,\iden-K\}$ where $K$ is the EPR-assisted AGSP corresponding to $\ket{\Omega}$ is lower bounded by 
 \begin{align}
     \CC_{\D}(\Omega)&\geq \ES_{2\cdot(2\D)^{2/3}}(\psi)-1= S_{\max}^{2\cdot(2\D)^{2/3}}(\Omega_A)-S_{\min}^{2\cdot(2\D)^{2/3}}(\Omega_A)-1.
 \end{align}
 \end{thm}

We note that the above theorem applies to any state $\ket{\Omega}$, as long as an approximate projection operator $K$ (similar to Definition \ref{def:ERP AGSP informal}) exists. But we keep the ``ground state'' terminology in our discussion, to fit the context.

\subsection{Communication protocol for approximate ground space projection}\label{sec:Communication protocol for approximate ground space projection}
 
In \thmref{upperbound on spread infromal}, we stated a lower bound on $\CC_{\D}(\Omega)$, the communication complexity of approximately measuring the ground state. In this section, we design a communication protocol that implements such a measurement and provides us with an \emph{upper bound} on the communication cost $\CC_{\D}(\Omega)$.  

Let $|A|$ be total number of qudits on Alice's side. Alice and Bob can trivially implement $K$ by exchanging $|A|$ qudits. Although, as we saw before, this bound can be tight for some states like $\frac{1}{\sqrt{2}}(\ket{00}^{\ot |A|}+\ket{\Phi}^{\ot |A|})$, our result shows that when the input state is the ground state of a gapped Hamiltonian, the communication complexity can be improved to $O(|\partial A|)$, where $|\partial A|$ is the number of terms in the Hamiltonian that act on both Alice and Bob's registers, see \fig{1}.

\begin{thm}[Communication protocol for projecting onto the ground space]
\label{thm:reflectaroundgs_informal}
Suppose the state $\ket{\Omega}$ is the ground state of a local Hamiltonian with spectral gap $\g$ (See \secref{prelim} for a formal definition of ``local Hamiltonian.''). Let $|\partial A|$ be the number of terms in the Hamiltonian that acts on both Alice and Bob's qudits. Then, there exists a protocol that implements the measurement $\left\{K,\iden-K\right\}$, where $K$ is an EPR-assisted AGSP satisfying \begin{align}
\norm{K(\ket{\Phi}\otimes \iden)-\ket{\Phi}\otimes\ketbra{\Omega}{\Omega}}\leq \Delta\nn
\end{align}
and has the communication cost
\begin{align}
c=O\left(\frac{|\partial A|}{\g} \log\left(\frac{|\partial A|}{\g \D}
\log \frac{1}{\D}\right) \log \frac{1}{\D}\right)\label{eq:t110}
\end{align}
\end{thm}
As a result of \thmref{upperbound on spread infromal} in the previous section, we know that the communication complexity of performing an AGSP gives us information about the distribution of the Schmidt coefficients in the ground state. When combined with the bound \eqref{eq:t110}, this establishes an ``area law'' for the entanglement spread, meaning that across a given cut in the ground state of a gapped Hamiltonian, the Schmidt coefficients can be spread out at most proportional to the size of the cut.
\begin{cor}[Area law for entanglement spread]\label{cor:Area bound on entanglement spread informal}
Under the assumptions of \thmref{upperbound on spread infromal} and  \thmref{reflectaroundgs_informal}, the following bound on the entanglement spread of the ground state $\ket{\Omega}$ of a gapped local Hamiltonian holds,
\begin{align}
          \ES_{\d}(\Omega)\leq O\left(\frac{|\partial A|}{\g}\cdot \log\left(\frac{|\partial A|}{\g \d^{3/2}}\log \frac{1}{\d^{3/2}}\right)\cdot \log \frac{1}{\d^{3/2}}\right)\label{eq:e1}
\end{align}
\end{cor}
Note that the range of applicability of   \corref{Area bound on entanglement spread informal} is quite general. The bound  \eqref{eq:e1} holds for any local Hamiltonian over an arbitrary interaction (hyper)graph. In particular, we do not assume the Hamiltonian is also \emph{geometrically local} or the qudits are arranged on a \emph{lattice}. In fact, when the Hamiltonian is restricted to any finite dimensional lattice, we can obtain tighter bounds on the entanglement spread by lifting the powerful machinery of AGSPs based on the Chebyshev polynomials \cite{area_law_subexponential_algo,one-dimensional_area_law} from 1D geometries to higher dimensions. This quadratically improves the bound \eqref{eq:t110} to $c=\tilde{O}(\sqrt{|\partial_w A|/\g})$ at the cost of including an extended boundary $\partial_w A$ of constant width $w$ instead of the original boundary $\partial A$ (see \secref{connection to the counter example} for a related discussion). Note that $|\partial_w A|=O(|\partial A|)$ on lattices when $w=O(1)$. In this setting, it is more natural to view $2^c$ as the Schmidt rank of the AGSP operator across the cut. We also do not rely on shared EPR pairs in this setup. To distinguish things from our previous construction, we refer to this operator as the Chebyshev-AGSP.  More precisely, we have:

\begin{thm}[Chebyshev-AGSP for lattices]\label{thm:Chebyshev-AGSP_informal}
Suppose, $H$ is a geometrically-local Hamiltonian with gap $\g$ over a finite-dimensional lattice. Let $(A:B)$ be a bipartition of the lattice. There is an operator $K$ with the Schmidt rank $2^c$ across the partition $(A:B)$ such that $\norm{K-\ketbra{\Omega}{\Omega}} \leq \Delta$ and
\begin{align}
% \label{eq:schmidtrank}
c= \tilde{O}\left(\sqrt{\frac{|\partial_{w} A|}{\gamma}}\cdot \log\left(\frac{ |\partial_w A|^2}{\g}\log^2(\frac{1}{\Delta})\right)\cdot \log\frac{1}{\Delta}\right),
\end{align}
where $\tilde{O}$ hides constant factors related to the geometry of the Hamiltonian and $w=O(1)$. Here, $|\partial_w A|$ is the number of terms in the Hamiltonian that act on the qudits in some extended boundary of constant width around $\partial A$.
\end{thm}
\begin{cor}[Tighter bounds on entanglement spread on lattices]\label{cor:Tighter bounds on entanglement spread on lattices}
Under the same conditions in \thmref{upperbound on spread infromal} and \thmref{Chebyshev-AGSP_informal}, the entanglement spread of the ground state of geometrically local Hamiltonians is bounded by 
\begin{align}
          \ES_{\d}(\Omega)\leq\tilde{O}\left(\sqrt{\frac{|\partial_w A|}{\gamma}}\cdot  \log\left(\frac{ |\partial_w A|^2}{\g}\log^2(\frac{1}{\d^{3/2}})\right)\cdot \log\frac{1}{\d^{3/2}}\right)\label{eq:e2}
\end{align}
where $w=O(1)$.
\end{cor}

What is the operational difference between these two approaches? If the Chebyshev-AGSP has Schmidt rank $2^c$ then it can be thought of as resulting from a {\em non-unitary} protocol that communicates $O(c)$ qubits.  Equivalently, it could be a unitary protocol that uses post-selection, meaning that it has some probability $<1$ of outputting ``don't know'' and conditioned on not answering ``don't know'' has a good chance of correctly distinguishing the ground state.  We do not know whether the better parameters of \corref{Tighter bounds on entanglement spread on lattices} can be achieved using unitary protocols or on general graphs.

\section{Main ideas}
Here, we describe the main ideas and technical tools used in the proof of our results. 
\subsection{AGSP from quantum phase estimation} 

One ingredient of our proofs is a novel construction of an AGSP for the ground state of gapped Hamiltonians based on the quantum phase estimation (QPE) algorithm. We find a protocol between Alice and Bob that allows them to jointly apply this AGSP with communication complexity $\tilde{O}(|\partial A|/\g)$. As mentioned in \secref{Communication protocol for approximate ground space projection}, one advantage of using QPE compared to the conventional Chebyshev polynomials (reviewed in \secref{high dim AGSp using Chebyshev}) is that it applies not only to geometrically-local Hamiltonians on lattices, but also continues to work for any local Hamiltonian on arbitrary interaction graphs.

One can view QPE as a procedure that given an eigenstate of a Hamiltonian $H$, uses $O(\log(1/\g))$ many ancillary qubits, makes $O(1/\g)$ queries to the Hamiltonian simulation oracle $e^{-iH}$, and determines the energy of the input state with accuracy $\g/2$. By letting $\g$ be the gap of the Hamiltonian, this algorithm basically performs a two outcome measurement $\{\Omega,\iden-\Omega\}$ on any input state, where $\Omega$ is the ground state of $H$.

To implement this algorithm in a distributed fashion involving two parties, Alice and Bob need to prepare and reflect about the state $\frac{1}{\sqrt{T+1}}\sum_{t=0}^{T}\ket{t}\ket{t}$ for $T=O(1/\g)$ and work together to apply the operator $e^{-itH}$ conditioned on the register $\ket{t}$. In the next section, we show how to achieve this.
\subsection{Communication protocol based on interaction picture Hamiltonian simulation} 
For a given partition of the qudits between Alice and Bob, we can write the Hamiltonian as $H=H_A+H_{\partial A}+H_B$ where $[H_A,H_B]=0$. One of our main technical contributions is designing a communication protocol for performing the Hamiltonian simulation operator $e^{-itH}$ with a communication cost that scales as $O(t\norm{H_{\partial A}})$ instead of the conventional $O(t\norm{H})$.

It is not hard to see how one can achieve this if the boundary term $H_{\partial A}$ also commutes with $H_A$ and $H_B$. In that case, we have $e^{-it H}=e^{-itH_A} e^{-itH_{\partial A}}e^{-itH_B}$ and the parties can implement $e^{-it H}$ if one of them sends the boundary qudits that are in the support of $H_{\partial A}$ to the other. This yields a communication cost that scales with $|\supp(H_{\partial A})|=O(|\partial A|)$. In general, however, $H_{\partial A}$ does not commute with $H_A$ and $H_B$ and finding a non-trivial protocol for the Hamiltonian simulation becomes challenging. 

One attempt to remedy this might be to use the Trotterization technique. That is, to divide the simulation into $\eta$ segments and implement $e^{-itH/\eta}$ for $\eta$ consecutive times. If $\eta$ is large enough, $[H_{\partial A}/\eta,H_{A \text{ or } B}/\eta]\approx 0$, and we again recover the commuting case. That is, the parties collaboratively implement $e^{-itH_{\partial A}/\eta}$. Unfortunately, for this to work, we need $\eta$ (and therefore, the communication cost) to be $O(t\norm{H})$, which is far from the bound $O(t \norm{H_{\partial A}})$ we are aiming for. 

We instead use a recent framework for Hamiltonian simulation developed in \cite{low2018Hamiltonian} known as the ``interaction picture'' Hamiltonian simulation. Intuitively, one can view this as a sophisticated change of variables that is widely used in physics and allows us to separate the contribution of the boundary term from $H_A$ and $H_B$. Suppose we want to prepare the state $\ket{\psi(t)}=e^{-itH}\ket{\psi(0)}$. For any $\ket{\psi(t)}$, we define its counterpart in the interaction picture by
\begin{align}
    \ket{\psi_I(t)}=e^{-it(H_A+H_B)}\ket{\psi(t)}.\label{eq;e8}
 \end{align}
Since the operator $e^{-it(H_A+H_B)}$ can be applied locally by the parties, the states $\ket{\psi_I(t)}$ and $\ket{\psi(t)}$ can be switched with each other with no extra communication. The point of this transformation is that the state $\ket{\psi_I(t)}$ can be prepared starting from $\ket{\psi(0)}$ by applying a unitary $U(t)$ which is the Hamiltonian simulation operator associated with a \emph{time dependent} Hamiltonian 
\begin{align}
H_I(t)=e^{it(H_A+H_B)}H_{\partial A}e^{-it(H_A+H_B)}.\label{eq:e12}
\end{align}
Putting the time-dependence of $H_I(t)$ aside (we discuss that in more details in \secref{Interaction picture protocol}), the main gain is that $\norm{H_I(t)} =\norm{H_{\partial A}}$. This solves the issue we mentioned before because here, the length of Trotter step $\eta$ in implementing $U(t)$ can be taken as small as $O(t\norm{H_{\partial A}})$ instead of the original $O(t\norm{H})$. The remaining task is to find a communication protocol for performing $U(t/\eta)$, which now, is not simply the operator $e^{-itH_{\partial A}/\eta}$ that we had before. This is done in \cite{berry2015simulating,low2018Hamiltonian} using the Linear Combination of Unitaries (LCU) method. Our next idea is a modification of this algorithm that suits our framework better. 

\subsection{EPR-assisted communication and the LCU method}
Our results regarding the ground state entanglement and the communication complexity are information theoretic in nature. In particular, the running time or other algorithmic aspects of the tools we use, such as the Hamiltonian simulation, do not affect our conclusions. Here, we explain how we can use this observation to simplify the analysis of a part of our protocol.

In the LCU method, one Taylor expands the Hamiltonian simulation operator $U(t)$ to get $U(t)\approx\sum_k \a_k u_k^{(A)} \ot u_k^{(B)}$ for some choice of coefficients $\a_k\in \bbR$ and unitaries $u_k^{(A)}$ and $u_k^{(B)}$ that act on Alice and Bob's qudits respectively. To keep the running time efficient such Taylor expansions are truncated at low orders. 

When Alice and Bob jointly implement the LCU algorithm, they need to prepare and share the ancillary state $\ket{\a}=(\sum_k \a_k)^{-1/2} \sum_k \sqrt{\a_k} \ket{k}_A\ket{k}_B$. Then, they proceed by applying the unitaries $u_k^{(A)} \ot u_k^{(B)}$ conditioned on their register $\ket{k}$. Now suppose instead of truncating the expansions, we continue adding higher terms. Of course, the issue is that the number of coefficients $\a_k$ and thus, the communication cost of sharing $\ket{\a}$ and reflecting about $\ket{\a}$ also  increases. On the other hand, we know that if instead of $\ket{\a}$, the parties share a maximally entangled state, the bound \eqref{eq:e1} on the entanglement spread remains intact. In other words, it is not the number of exchanged ancillary qudits in the protocol, but their entanglement spread that affects our final bound \eqref{eq:e1}. 

We fix this problem by modifying the LCU algorithm such that instead of the state $\ket{\a}$, Alice and Bob only share the maximally entangled state (or equivalently some number of EPR pairs).   This state only needs to be shared once, which can be done outside the protocol, and then many reflections about it can be done with a cost independent of the size of the state.  Now we can keep an unbounded number of terms in the expansions and avoid similar approximations in our protocol. This blows up the running time of these procedures, but maintains the communication complexity.

\subsection{AGSP for lattices}\label{sec:high dim AGSp using Chebyshev}
Our improved bound for the lattice Hamiltonians in  \thmref{Chebyshev-AGSP_informal} are obtained using the AGSPs based on the Chebyshev polynomials. These were first developed in the context of the area law for entanglement entropy in 1D systems \cite{one-dimensional_area_law,area_law_subexponential_algo,arad_1d}. The AGSP framework \cite{detectability_Aharonov, one-dimensional_area_law} in itself provides a framework to connect the min-entropy and entanglement entropy (see \cite[Lemma 5.3]{detectability_Aharonov} or \cite[Lemma III.3]{one-dimensional_area_law}). But this connection does not give us the desired bound on entanglement spread, as it relates entanglement entropy and min-entropy by a certain \emph{multiplicative factor}, that may be large. For instance, \cite[Lemma III.3]{one-dimensional_area_law} implies that by choosing the Chebyshev-based AGSP which has a shrinking of $O(1)$ and the Schmidt rank of $2^{O(\sqrt{|\partial A|})}$, we get 
\begin{equation}
\label{eq:minentropytoentropy}
S(\Omega_A)= O(\sqrt{|\partial A|}S_{\min}(\Omega_A)),
\end{equation}
where $S(\Omega_A)$ is the von-Neumann entropy of $(\Omega_A)$.     

Here, we show that a simple adaptation of the Chebyshev-based AGSP, along with appropriate smoothing, leads to a stronger theorem for lattices, which shows that entanglement spread scales as $ O(\sqrt{|\partial A|})$ (see discussion section for the interpretation). We utilize the ``truncation step'' \cite{area_law_subexponential_algo} which is used to lower the norm of the Hamiltonian away from a cut while maintaining its gap and ground state. We apply the truncation to both the frustration-free and frustrated cases. In the former, we use the Detectability Lemma operator \cite{detectability_Aharonov}, while in the latter, we rely on the recent techniques of \cite{Kuwahara_area_law19} to perform the truncation. 
\section{Discussion and connection to previous work}\label{discussion}
\paragraph{Quadratically better scaling on lattices:} In \corref{Tighter bounds on entanglement spread on lattices}, we have shown that the entanglement spread on lattices scales as $\sqrt{|\partial A|}$ The intuition behind this comes from the exponential decay of correlations which is shown to hold for gapped Hamiltonians on any finite dimensional lattice \cite{Hastings04,HastingsK06, NachtergaeleS06}. The decay of correlations implies that the distant qudits along the boundary $\partial A$ are almost uncorrelated. This suggests that the ground state across the boundary is roughly in a product form $\ket{\phi}^{\ot |\partial A|}_{AB}$ composed of $O(|\partial A|)$ partially entangled states $\ket{\phi}_{AB}$ . By using conventional concentration bounds  \cite{PopescuConcentration}, it can be shown that the \emph{smooth} entanglement spread obeys $\ES_{\d}(\phi^{\ot k})= O(\sqrt{k})$, which is quadratically smaller than the $\d=0$ case where $\ES(\phi^{\ot k})=O(k)$. Thus, an entanglement spread of $O(\sqrt{|\partial A|})$ that we prove for gapped lattice Hamiltonians matches our intuitive expectation. 

One might wonder if our quadratic bound in \thmref{Chebyshev-AGSP_informal} for lattices can be improved. Here, we show that this is not possible in general. Consider a $D$-dimensional cubic lattice $[n]^{D}$ for an even $n$ such that the qubits are located on the vertices of the lattice. Let $A=[\frac{n}{2}-1]\times [n]\times \ldots \times [n]$ define a bipartition of this lattice. Suppose, we have a Hamiltonian $H$ on the lattice given by 
$$H=\sum_{i_1, \ldots, i_D: i_1=\text{even}} (\iden - \Psi_{i_1, \ldots i_D}),$$ 
where the entangled states 
$$\ket{\Psi}_{i_1, \ldots i_D}=\left(\sqrt{\frac{2}{3}}\ket{0}_{i_1, \ldots i_D}\ket{0}_{i_1+1, \ldots i_D}+\sqrt{\frac{1}{3}}\ket{1}_{i_1, \ldots i_D}\ket{1}_{i_1+1, \ldots i_D}\right)$$
is defined between qudits $(i_1,i_2,\ldots i_D)$ and $(i_1+1,i_2,\ldots i_D)$. Then, the ground state is the simple two-qudit product state $\bigotimes_{i_1, \ldots i_D: i_1=\text{even}}\ket{\Psi}_{i_1, \ldots i_D}$. It is easily seen that the entanglement spread across this partition is at least $\Omega(\sqrt{n^{D-1}})=\Omega(\sqrt{|\partial A|})$ achieving the bound in \thmref{Chebyshev-AGSP_informal}.

Is it possible to prove an analog of Theorem \ref{thm:Chebyshev-AGSP_informal} with $S_{\min}^{2(2\D)^{2/3}}(\Omega_A)$ replaced by $S_{\min}(\Omega_A)$? This cannot be done without changing the upper bound from $O(|\partial A|^{1/2})$ to $O(|\partial A|)$, since the two-qubit product ground state constructed above has the property that $S_{\min}(\Omega_A) \leq |\partial A|/4$ and $S_{\max}^{2(2\D)^{2/3}}(\Omega_A)\geq 0.92|\partial A|$. Improving $S^{2(2\D)^{2/3}}_{\max}(\Omega_A)$ to $S_{\max}(\Omega_A)$ is also not possible since there are Hamiltonians such as the transverse field Ising model that are gapped but have $S_{\max}(\Omega_A)$ scaling with the system size \cite{calabrese2004entanglement}.

\paragraph{Li-Haldane conjecture on the entanglement spectra:} An application of our result in \thmref{Chebyshev-AGSP_informal} and \corref{Tighter bounds on entanglement spread on lattices} is to give formal evidence for the Li-Haldane conjecture. According to this conjecture, the entanglement spectrum of a gapped ground state over a cut is similar to the spectrum of a Hamiltonian, known as the modular Hamiltonian, that acts only on the boundary \cite{LiHaldane:2008} \footnote{Note that assuming this conjecture, we see that the entanglement entropy of the reduced ground state is close to that of the boundary Gibbs state and therefore, obeys an area law.}. This a surprising fact given that in general, the spectrum of a mixed state over a region $A$ has a support of size $e^{O(|A|)}$. Whereas according to this conjecture for gapped Hamiltonians, the spectrum of the reduced ground state is similar to that of a very specific state -- i.e. the Gibbs state -- acting on a much smaller space $\partial A$. To see this connection, we should look more closely at the distribution of the eigenvalues of the Gibbs state. To this end, we use recent results that establish concentration bounds on the energy distribution of this state. More formally, \cite[Corollary 1]{Kuwahara_Saito_gibbs_concentration_2019} shows that for a Gibbs state $\rho_{\beta}=\frac{e^{-\beta \modH}}{\Tr[e^{-\beta \modH}]}$ of a local Hamiltonian $\modH$ with $\beta=O(1)$, it holds that $$\Tr[\Pi_{[\xi_-,\xi_+]}\rho_{\beta}]\geq 1-\d,$$ where $\Pi_{[\xi_-,\xi_+]}$ is the projector onto the states with energy in the range $[\xi_-,\xi_+]$, where
\begin{align}
\xi_{\pm}=\Tr[\modH\rho_{\beta}]\pm  O(1)\cdot \sqrt{\|\modH\|\log\frac{1}{\d}}.\nn
\end{align}
In other words, the spectrum of the Gibbs state $\r_{\b}$ is concentrated in the energy range $\pm O\left(\norm{\modH}^{1/2}\right)$ around the average energy $\Tr[\modH \r_{\b}]$. This result implies that the entanglement spread of the spectrum of the Gibbs state satisfies
\begin{align}
\ES_{\d}(\rho_{\beta}) \leq \log \left( \frac{e^{-\b \xi_{-}}}{e^{-\b \xi_{+}}}\right)=O(1)\cdot \sqrt{\|\modH\|\log\frac{1}{\d}}.\nn
\end{align}
Since based on our assumption, $\modH$ is a modular Hamiltonian with support only on the boundary of a region $A$, it has the norm $\|\modH\|=O(|\partial A|)$. Thus, the entanglement spread of its Gibbs state scales as $\ES_{\d}(\rho_{\beta})=O(1)\sqrt{|\partial A|\log\frac{1}{\d}}$. This is noticeably similar to the scaling predicted by our result for the gapped ground states in \corref{Tighter bounds on entanglement spread on lattices}. This shows that at least from the perspective of the entanglement spread, the reduced ground state behaves similar to a Gibbs state supported on the boundary.
\paragraph{Implications for proving area law for the entanglement entropy:} Theorem \ref{thm:Chebyshev-AGSP_informal} shows that if one can prove an area law for $S_{\min}^{2(2\D)^{2/3}}(\Omega_A)$, then this implies an area law for $S_{\max}^{2(2\D)^{2/3}}(\Omega_A)$ and hence, for the entanglement entropy. In contrast, prior work (see Equation \ref{eq:minentropytoentropy}) would show that an area law for the min-entropy $S_{\min}(\Omega_A)=O(|\partial A|)$ leads to a sub-volume law $O(|\partial A|^{3/2})$ on the entanglement entropy. We cannot directly compare our result with this, as $S_{\min}(\Omega_A)$ is smaller than $S_{\min}^{2(2\D)^{2/3}}(\Omega_A)$. 

As mentioned earlier, we cannot replace $S^{2(2\D)^{2/3}}_{\min}(\Omega_A)$ with $S_{\min}(\Omega_A)$ without changing our upper bound on the entanglement spread to $O(|\partial A|)$. Achieving such a bound is an interesting open problem since it would rigorously prove that min-entropy area law implies entanglement entropy area law. The utility of this is that min-entropy area law may be easier to prove in comparison to the entanglement entropy area law. For instance, for specific models such as stoquastic local Hamiltonians, proving min-entropy area law can be reduced to a classical problem \cite{Bravyi_stoquastic}

\paragraph{Connection to the counter example to the area law in \cite{Aharonov2014CounterExampleToAreaLaw}}\label{sec:connection to the counter example}
Our setup is closest to  \cite{Aharonov2014CounterExampleToAreaLaw}, where the authors construct a family of gapped Hamiltonians whose ground states violate the entropy area law. This is done by connecting a protocol for testing maximally entangled states to the ground state of a local Hamiltonian using Kitaev's circuit-to-Hamiltonian construction. The obtained ground state admits a bipartition into parts $A$ and $B$ such that a single Hamiltonian term crosses the cut, but it enforces a maximally entangled state $\Phi_{|A|}$ between $A$ and $B$. This causes the entropy of part $A$ to be $O(|A|)$ violating the entropy area law. Nevertheless, we see that this construction still satisfies our entanglement spread area law simply because the maximally entangled state has zero entanglement spread, and the ground state has entanglement spread at most $O(1)$. When combined with our previous discussion on area law, this loosely suggests the following: The ground state of a gapped Hamiltonian always exhibits a small entanglement spread. But it either has a large min-entropy (such as maximally entangled states in the counter-example Hamiltonian) hence not obeying an entropy area law, or it possesses small min-entropy (such as the 1D ground states) thus obeying an entropy area law.

\paragraph{Hamiltonian simulation by Trotterization:} While we use the interaction picture Hamiltonian simulation algorithm, it would be interesting to achieve the same result by directly using the Trotterization method along the lines of \cite{berry2007nonlocalSimulation}. That is, we want to simulate the Hamiltonian $e^{-i\tau H_{\partial A}}$ for some small $\tau\approx 1/\norm{H}$ with the communication cost $O(\tau\norm{H_{\partial A}})$. By repeating this step for $1/\tau= O(\norm{H})$ times, we obtain the desired overall scaling of $O(\norm{H_{\partial A}})$. The issue with naively using this approach is that each simulation step requires exchanging one qubit of communication resulting in a large communication complexity. Note, however, that the \emph{entropy} of this exchanged qubit is $O(\tau)$. Hence, we expect the quantum information cost of this step \cite{Touchette:2015} to also be $O(\tau)$. We anticipate that performing quantum information theoretic \emph{compression} on such a protocol would lead to a new protocol achieving the desired bound. Finally, note that recent results in \cite{childs2019theoryoftrotter} achieve a similar bound for the clustered Hamiltonians by a tighter analysis of the Trotter error.

\paragraph{Compression of Schmidt rank:} A by-product of our techniques is a compression tool for the Schmidt rank of any AGSP using EPR assistance. Since this might find other applications beyond our work, we formally state it in the following proposition.

\begin{prop}
\label{prop:compressionlemma}
Fix an AGSP
$$K= \sum_i \alpha_i U_i\otimes V_i,$$
with $\alpha_i>0$, $\|U_i\|,\|V_i\|\leq 1$, and $U_i$, $V_i$ acting on subsystem $A$, $B$ respectively. Suppose 
$$\norm{K-\ketbra{\Omega}{\Omega}}\leq \Delta.$$
Then there exists an EPR-assisted AGSP $K'$ (as in \defref{ERP AGSP informal}) with Schmidt rank $\left(\frac{\sum_i\alpha_i}{\Delta}\right)^{\mathcal{O}(1)}$ such that $$\norm{K'(\ket{\Phi}\otimes \iden)-\ket{\Phi}\otimes\ketbra{\Omega}{\Omega}}\leq 2\Delta.$$
\end{prop}
Hence, given an AGSP $K$, we can use \propref{compressionlemma} to construct an EPR-assisted AGSP with similar shrinking $\D$ but a Schmidt rank only polynomial in $\ell_1$-norm of the coefficients $\a_i$ in the Schmidt decomposition of $K$.

\section{Preliminaries}\label{sec:prelim}

\paragraph{Local Hamiltonians:} Let $\cS$ be a collection of $n$ spins, each with dimension $s$. The interactions between these spins are described by a local Hamiltonian $H=\sum_{k=1}^N h_k$
where the operators $0 \preceq h_k \preceq \iden$ act nontrivially only on at most $\k$ spins. Let $H_X$ denote the Hamiltonian restricted to region $X\subseteq \cS$. For a bipartition $(A : B)$ of the set $\cS$, we write $H=H_A+H_{B}+H_{\partial A}$, where $H_{\partial A}$ is the collection of interaction terms acting on both $A$ and $B$. We denote the Hilbert space of these partitions and the whole system by $\cH_A$, $\cH_B$ and $\cH_{AB}$ respectively.

We denote the spectrum of $H$ by $E_0,E_1,\dots, E_{\max}$. Let $\ket{\Omega}$ be the \emph{unique} ground state of $H$ and $\ket{E_1},\ket{E_2},\dots,\ket{E_{\max}}$ the other eigenstates. The spectral gap of the Hamiltonian $H$ is a constant $\g$ such that $E_1=E_0+\g$.

\paragraph{Communication protocols:} In what follows, we consider quantum communication protocols between Alice and Bob. We assume, a bipartition $(A : B)$ of the set $\cS$ is shared between the parties such that Alice has access to spins in  region $A$ while Bob has access to those in region $B$. Both parties also have their own additional registers.

The parties communicate by sending qubits, and can cooperate to implement an operator supported on $A \cup B$. The communication complexity of implementing such an operator is defined as the total number of exchanged qubits. 

\paragraph{Two-party entanglement}
Given a state $\ket{\psi}_{AB}$ shared between Alice and Bob with the Schmidt coefficients $\l_1\geq \l_2\geq \dots\geq \l_d$, the R\'enyi entropy of order $\a$ of the reduced state $\psi_A=\Tr_{B}\ketbra{\psi}{\psi}$ is defined as 
\begin{align}
S_{\a}(\psi_A)=\frac{1}{1-\a}\log \left(\Tr \psi_A^{\a}\right)=\frac{1}{1-\a}\log \left(\sum_{i=1}^d \l_i^{\a}\right), \quad 0<\a<\infty.
\end{align}
Specifically for $\a=0,1,\infty$, we define the $\max$- and $\min$- entropies by $S_{\max}(\psi_A)=\log\left(\rank\left(\psi_A\right)\right)$ and $S_{\min}(\psi_A)=-\log \l_1$. The von Neumann entropy $S(\psi_A)$ is the limiting case of $\lim_{\a\rightarrow 1}S_{\a}(\psi_A)=-\Tr[\psi_A \log \psi_A]$. 
In this paper, we mostly use a robust version of these entropies defined as follows. 
\begin{definition}[Smooth R\'enyi entropies and entanglement spread]\label{def:smooth ent}
 Consider a state $\r$ with eigenvalues $\l_1\geq \l_2 \geq \dots \geq \l_d$. For $\d \in (0,1)$, let
 \begin{align}
 r_{\d}(\r)=\{L\subseteq [d]:\sum_{i\in L} \l_i\geq 1-\d\}.
 \end{align}
 We define the $\d$-smooth max- and min- entropies of the state $\r$ by 
 \begin{align}
     S_{\max}^{\d}(\r)&=\min_{L\in r_{\d}(\r)} \log |L|,\label{eq:k4}\\ 
     S_{\min}^{\d}(\r)&=-\min_{L\in r_{\d}(\r)} \log \left(\max_{i\in L} \l_i\right).
 \end{align}
The $\d$-smooth entanglement spread of the state $\r$ is defined as
\begin{align}
    \ES_{\d}(\r)=S_{\max}^{\d}(\r)-S_{\min}^{\d}(\r)
\end{align}
\end{definition}
\begin{lem}[Young-Eckart theorem]\label{lem:Young-Eckart}
Consider a bipartite state $\ket{\psi}\in \cH_{AB}$ with the Schmidt coefficients $\l_1\geq \l_2\geq \dots\geq \l_d$. Let $\ket{\phi}\in \cH_{AB}$ be the state with the Schmidt rank $\leq r$ which has the largest overlap with $\ket{\psi}$. It holds that $|\braket{\phi}{\psi}|^2\leq \sum_{i=1}^r \l_i$.
\end{lem}

\section{Approximate ground space projector and entanglement spread}
\begin{definition}[EPR-assisted AGSP, restatement of \defref{ERP AGSP informal}]\label{def:AGSP}
Fix a bipartition $(A:B)$ of the spins, let $\ket{\Omega}\in \cH_{AB}$ be the ground state of a local Hamiltonian $H$ and $\ket{\Phi}=\frac{1}{\sqrt{p}}\sum_{j=1}^{p} \ket{j}_A\ket{j}_B$ be a maximally entangled state with an arbitrarily large dimension $p$ shared across $(A:B)$. We say that an operator $K$ is a $(D,\D)$-EPR-assisted AGSP if
\begin{itemize}
    \item[--] The Schmidt rank of $K\ket{\Phi}\ket{\psi}$ is at most $D$ times the Schmidt rank of $\ket{\Phi}\ket{\psi}$, for any bipartite state $\ket{\psi}$ across $(A:B)$, and
    \item[--] It holds that
    \begin{align}
\norm{K(\ket{\Phi}\otimes \iden)-\ket{\Phi}\otimes\ketbra{\Omega}{\Omega}}\leq \Delta,\label{eq:k1}
\end{align}
\end{itemize}
\end{definition}
\begin{rem}
The Schmidt rank of the AGSP equals $2^c$ where $c$ is the communication complexity of implementing it. We switch between $D$ and $c$ where ever it is more convenient to use one. Also, we use EPR-assistance only in our AGSP construction based on the quantum phase estimation and not the Chebyshev-AGSPs. Nevertheless, the following theorem applies generally to both cases. 
\end{rem}

\begin{thm}[Bounding entanglement spread using AGSP]\label{thm:bound ES using AGSP}
\label{thm:AGSPtospread}
Suppose there exists a $(D,\D)$-EPR-assisted AGSP with respect to a partition $(A:B)$ such that $\D<\frac{1}{4\sqrt{2}}$. Then the entanglement spread across $(A:B)$ is bounded by
\begin{align}
S_{\max}^{2(2\Delta)^{2/3}}(\Omega_A) - S_{\min}^{2(2\Delta)^{2/3}}(\Omega_A) \leq  \log D + 1.\label{eq:k5}
\end{align}
\end{thm}
\begin{proof}
Let $\e=(2\D)^{2/3}$. Consider the Schmidt decomposition
$$\ket{\Omega}= \sum_i\sqrt{\lambda_i}\ket{i}_A\ket{i}_{B},$$
where $\{\lambda_i\}$ are in descending order. Let $b$ be the smallest integer such that $\varepsilon':=\sum_{i<b}\lambda_i\geq \varepsilon$ and define
$$\ket{\Omega_{\text{heavy}}}=\sum_{1\leq i<b}\sqrt{\frac{\lambda_i}{\varepsilon'}}\ket{i}_A\ket{i}_{B},$$
$$\ket{\Omega_{\text{light}}}=\sum_{i\geq b}\sqrt{\frac{\lambda_i}{1-\varepsilon'}}\ket{i}_A\ket{i}_{B}.$$
Hence, $\ket{\Omega}=\sqrt{\e'}\ket{\Omega_{\text{heavy}}}+\sqrt{1-\e'}\ket{\Omega_{\text{light}}}$. Let $\ket{\Phi}$ have the Schmidt decomposition 
$$\ket{\Phi}=\frac{1}{\sqrt{p}}\sum_{j=1}^p\ket{j}_A\ket{j}_{B},$$ for some integer $p$.
From the closeness of $K$ to $\ket{\Omega}$ as in \eqref{eq:k1} and the identity $\braket{\Omega}{\Omega_{\text{heavy}}}=\sqrt{\varepsilon'}$, we have
\begin{equation}
\label{eq:closetogr}
\|K\ket{\Phi}\ket{\Omega_{\text{heavy}}}-\ket{\Phi}\otimes\ket{\Omega}\braket{\Omega}{\Omega_{\text{heavy}}}\|=\|K\ket{\Phi}\ket{\Omega_{\text{heavy}}}-\sqrt{\varepsilon'}\ket{\Phi}\ket{\Omega}\|\leq \Delta.
\end{equation}
The Schmidt rank of $\frac{1}{\sqrt{\varepsilon'}}K$ is the same as $K$ which equals $D$. The Schmidt rank of $\ket{\Omega_{\text{heavy}}}$ is $b-1$. Hence, Schmidt rank of $\frac{1}{\sqrt{\varepsilon'}}K\ket{\Phi}\ket{\Omega_{\text{heavy}}}$ is at most $p(b-1)D$. From \eqref{eq:closetogr}, we have that
\begin{align}
\left|\bra{\Phi}\bra{\Omega}\frac{\left(\frac{1}{\sqrt{\varepsilon'}}K\right)\ket{\Phi}\ket{\Omega_{\text{heavy}}}}{\Norm{\left(\frac{1}{\sqrt{\varepsilon'}}K\right)\ket{\Phi}\ket{\Omega_{\text{heavy}}}}}\right|\geq 1-2\frac{\Delta}{\sqrt{\varepsilon'}}.
\end{align}
Following \cite{one-dimensional_area_law}, we use the Young-Eckart theorem (\lemref{Young-Eckart}) along with the above bound. This implies that the sum of the largest $p(b-1)D$ eigenvalues of $\Omega_A\otimes \Phi_A$ is at least $(1-\frac{2\Delta}{\sqrt{\varepsilon'}})^2$. However, since the eigenvalues of $\Phi_A$ are all equal to $\frac{1}{p}$, this sum is equal to $\sum_{i=1}^{(b-1)D}\lambda_i$.  This is the key point in our proof where we use the fact that $\ket{\Phi}$ is maximally entangled; replacing it with a different state, such as an embezzling state, would cause this step to fail. Hence, we have  $$\sum_{i=1}^{(b-1)D}\lambda_i \geq (1-\frac{2\Delta}{\sqrt{\varepsilon'}})^2\geq 1-\frac{4\Delta}{\sqrt{\varepsilon'}}.$$
From the definition of the smooth max-entropy \eqref{eq:k4}, we see that this statement is equivalent to
\begin{align}
S_{\max}^{\frac{4\Delta}{\sqrt{\varepsilon'}}}(\Omega_A) \leq \log D +  \log(b-1).\nn
\end{align}
Since $S_{\max}^{\frac{4\Delta}{\sqrt{\varepsilon}}}(\Omega_A)\leq S_{\max}^{\frac{4\Delta}{\sqrt{\varepsilon'}}}(\Omega_A)$, we conclude that
\begin{align}
S_{\max}^{\frac{4\Delta}{\sqrt{\varepsilon}}}(\Omega_A) \leq \log D +  \log (b-1).\nn
\end{align}
Now, consider the following two cases:
\begin{enumerate}
\item[1)] {\bf $\varepsilon\geq \lambda_1$:} From the definition of $b$, $$\varepsilon'=\sum_{i<b}\lambda_i \leq \varepsilon + \lambda_b \leq 2\varepsilon.$$ 
Since ${\l_i}$ are arranged in the descending order, we also have $\e'\geq (b-1)\lambda_{b-1}$. This implies 
$$\log(b-1) \leq \log \varepsilon' + \log\frac{1}{\lambda_{b-1}}\leq \log (2\varepsilon) + \log\frac{1}{\lambda_{b}}.$$ By \defref{smooth ent}, $\log\frac{1}{\lambda_{b}} = S_{\min}^{\e'}(\Omega_A) \leq S_{\min}^{2\e}(\Omega_A)$. From this, we conclude that
$$S_{\max}^{\frac{4\Delta}{\sqrt{\varepsilon}}}(\Omega_A) - S_{\min}^{2\varepsilon}(\Omega_A) \leq  \log D + \log (2\varepsilon) \leq \log D + 1.$$
\item[2)] {\bf $\varepsilon < \lambda_1$:} In this case, $b=2$. Thus, $S_{\max}^{\frac{4\Delta}{\sqrt{\varepsilon}}}(\Omega_A) \leq \log D.$ Since $S_{\min}^{2\varepsilon}(\Omega_A)\geq 0$, we have  
$$S_{\max}^{\frac{4\Delta}{\sqrt{\varepsilon}}}(\Omega_A) - S_{\min}^{2\varepsilon}(\Omega_A)\leq \log D \leq \log D +1.$$
\end{enumerate}
By plugging in the value of $\e=(2\D)^{2/3}$, we arrive at \eqref{eq:k5} which concludes the proof.
\end{proof}

\thmref{bound ES using AGSP} implies that we can bound the entanglement spread in the ground state by finding an appropriate AGSP. In the next sections, we achieve this using two distinct approaches. First in \secref{AGSP form HS}, we use the phase estimation algorithm to construct an AGSP for a gapped Hamiltonian on an arbitrary graph with $D=O(|\partial A|/\g)$. Next in \secref{AGSP for lattices}, we find an AGSP using the Chebyshev polynomial with a quadratically improved scaling of $D=O(\sqrt{|\partial A|/\g})$. 
\section{AGSP for general graphs using quantum phase estimation}\label{sec:AGSP form HS}
In this section, we describe a communication protocol between Alice and Bob that allows them to approximately implement the evolution operator $e^{-itH}=e^{-it(H_A+H_{B}+H_{\partial A})}$ using $\tilde{O}(\norm{H_{\partial A}} t)$ qubits of communication. The conventional Hamiltonian simulation techniques work in the Schr\"odinger picture. Naively using these techniques results in  communication complexity that scales with $\norm{H}$ instead of $\norm{H_{\partial_A}}$. To get around this issue, we instead use the recent Hamiltonian simulation algorithm in the interaction picture \cite{low2018Hamiltonian} along with the Linear Combination of Unitaries (LCU) method \cite{berry2015simulating}.
\subsection{Hamiltonian simulation in the interaction picture}\label{sec:Interaction picture protocol}
 In the Hamiltonian simulation, the goal is to prepare the state $\ket{\psi(t)}=e^{-it(H_A+H_{B}+H_{\partial A})}\ket{\psi(0)}$ for any initial state $\psi(0)$. This is conventionally done by directly implementing the unitary $e^{-it(H_A+H_{B}+H_{\partial A})}$. In the interaction picture, we work in the rotating frame $\ket{\psi_I(t)}=e^{it(H_A+H_{B})} \ket{\psi(t)}$. There, the evolution of a time-independent Hamiltonian $H$ is transformed to the evolution by a time-dependent Hamiltonian \begin{align}
 H_I(t)=e^{it(H_A+H_{B})}H_{\partial A}e^{-it(H_A+H_{B})}.\label{eq:t1}
 \end{align} 
We can divide the evolution of duration $t$ to $L$ shorter segments of length $\tau=t/L$. The state $\ket{\psi(t)}$ can be expressed in this picture by 
 \begin{align}
     \ket{\psi(t)}=e^{-it(H_A+H_{B})}\ket{\psi_I(t)} &=\left(e^{-i\tau(H_A+H_{B})} \cT\left[e^{-i\int_0^{\tau} H_I(s)ds}\right]\right)^L\ket{\psi(0)},\label{eq:t2}
 \end{align}
where $\cT\left[\exp\left({-i\int_0^{\tau} H_I(s)ds}\right)\right]$ is the time-ordered propagator. One advantage of working in the interaction picture is that $\norm{H_I(t)}=\norm{H_{\partial A}}$. Hence, the cost of implementing the propagation operator $\cT\left[\exp\left({-i\int_0^{\tau} H_I(s)ds}\right)\right]$ scales with $\norm{H_{\partial A}}$ instead of $\norm{H_{A}}$.

\begin{lem}[cf. \cite{low2018Hamiltonian}, Lemma 5]\label{lem:expansion for the time-ordered op}

The time-ordered propagator can be written as 
\begin{align}
\cT\left[e^{-i\int_0^{\tau} H_I(s)ds}\right]=\lim_{M,K\rightarrow \infty}\sum_{k=0}^{K}\frac{(-i\tau)^k}{M^k} \sum_{0\leq m_1 <\dots < m_k<M} H_I(m_k \tau/M)\cdots H_I(m_1 \tau/M).\label{eq:t4}
\end{align}
\end{lem} 
The order of the $M,K$ limit and the speed of convergence will not matter to us since we will see that our communication cost is completely independent of $M,K$.

The boundary term $H_{\partial A}$ can be decomposed as a sum of unitary operators, i.e., $H_{\partial A} = \sum_{j=1}^{J} \beta_{j} u^{(A)}_{j} \ot u^{(B)}_{j}$, where $u^{(A)}_{j}$ ($u^{(B)}_{j}$) acts on Alice's (Bob's) spins. We can always absorb the phase of $\b_j$ in $u^{(A)}_j$ and assume $\b_j>0$. Similarly, the interaction Hamiltonian is
\begin{align}
    H_I(m_k \tau/M)=\sum_{j =1}^{J} \beta_{j }(e^{im_k H_A\tau/M}u^{(A)}_{j }e^{-im_k H_A\tau/M}) \ot (e^{im_k H_B\tau/M  }u^{(B)}_{j }
    e^{-im_k H_B\tau/M })\nn.
\end{align}
By plugging this into \eqref{eq:t4}, we see that the time-ordered propagator can be expressed as a linear combination of unitary operators. For convenience, we define a collective index set 
\begin{align}
I_{M,K}=\{(k,m_1,\dots,m_k,j_1,\dots,j_k):  0 \leq k \leq K, 0\leq m_1<\dots<m_k<M, j_1,\dots,j_k\in [J] \}\nn.
\end{align} 
For some $\ell=(k,m_1,\dots,m_k,j_1,\dots,j_k)$, define 
\begin{gather}
\a_{\ell}=(\tau^k/M^k) \b_{j_k}\dots \b_{j_1}, \nn\\
v^{(A)}_{\ell}=(-i)^k(e^{im_k H_A \tau/M }u^{(A)}_{j_k}e^{-im_k H_A \tau/M })\dots (e^{im_1 H_A\tau/M}u^{(A)}_{j_1}e^{-im_1  H_A\tau/M}),\nn\\
v^{(B)}_{\ell}= (e^{im_kH_B \tau/M }u^{(B)}_{j_k}e^{-im_k H_B\tau/M })\dots (e^{im_1 H_B\tau/M }u^{(B)}_{j_1}e^{-im_1 H_B\tau/M }).\nn
\end{gather}

Note that 
\begin{equation}
\label{eq:normupbound}
\lim_{M,K\rightarrow \infty}\sum_{\ell \in I_{M,K}}\a_{\ell}= \exp(\tau \sum_{j=1}^J \b_j).    
\end{equation}
Using this notation, \lemref{expansion for the time-ordered op} can be expressed as
\begin{align}
    \cT\left[e^{-i\int_0^{\tau} H_I(s)ds}\right]=\lim_{M,K\rightarrow \infty}\sum_{\ell \in I_{M,K}}\a_{\ell}\ v^{(A)}_{\ell} \ot v^{(B)}_{\ell}.\label{eq:t6}
\end{align}
\subsection{Communication protocol for Hamiltonian simulation}

Here, we show how Alice and Bob can implement the evolution operator \eqref{eq:t6} using unlimited shared EPR pairs and quantum communication. First, we see how using $O\left( \log(1/\e) \right)$ qubits, they can perform a reflection about the maximally entangled state $\ket{\Phi_p}=\frac{1}{\sqrt{p}} \sum_{j=1}^p \ket{j}\ket{j}$ with an arbitrary dimension $p$ up to an error $\e$. To show this, we slightly modify the EPR testing protocol (i.e. performing the two-outcome measurement $\{\Phi_p,\iden-\Phi_p\}$) of \cite{Aharonov2014CounterExampleToAreaLaw}.

\begin{thm}[Reflection about $\ket{\Phi_p}$] \label{thm:protocol for reflection around the MES}
For any $p$ and any $\e>0$, there exists a protocol for performing $\iden-2 \Phi_p$, the reflection about the maximally entangled state, using $O\left(\log(1/\e)\right)$ qubits of communication. 
\end{thm}

\begin{proof}
Following \cite{Aharonov2014CounterExampleToAreaLaw}, we use quantum expanders to approximately construct $2\Phi_p-\iden$. A set of unitary operators $\{U_1,U_2,\dots U_{d}\}$ with each $U_j \in L(\bbC^p)$ is a $(p,d,\e)$ quantum expander if 
\begin{align}
\NORM{\frac{1}{d}\sum_{j=1}^{d} U_j \ot U^*_j-\Phi_p} \leq \e.\label{eq:t8}
\end{align}
There are constructions of quantum expanders that are independent of $p$ and achieve error $\e$ with $\log d=O\left(\log(1/\e)\right)$ \cite{hastings2007expander,harrow2007expander,harrow2009expander2}. Such quantum expanders can be used to reflect about $\Phi_p$ as in the following protocol with communication cost $2d$. 

\begin{enumerate}
    \item[0.] Alice and Bob share a state $\ket{\Psi}$ in registers $A,B$. Let $a$ be an ancillary register that Alice and Bob exchange. Their goal is to apply $(2\Phi_p-\iden)$ on $\ket{\Psi}$.
    \item[] \emph{First, they prepare $\frac{1}{\sqrt{d}}\sum_{j=1}^d \ket{j}\ U_j\ot U_j^*\ket{\Psi}$ in steps $1$ and $2$:}

    \item Alice prepares the state $\ket{s}=\frac{1}{\sqrt{d}}\sum_{j=1}^d \ket{j}$ in register $a$. She then performs $V=\sum_{j=1}^d\ketbra{j}{j}\ot U_i$ on her registers $a$ and $A$. Next, she sends register $a$ to Bob.
    \item Bob applies $V^*=\sum_{j=1}^d\ketbra{j}{j}\ot U^*_j$ on registers $a$ and $B$.
    \item Bob performs the reflection operator $2\ketbra{s}{s}-\iden$ on register $a$.
    \item[] \emph{Then they uncompute steps 1 and 2:}
    \item Bob applies $V$ on his registers and sends $a$ back to Alice.
    \item Alice performs $V^*$ on her registers and discards register $a$.
\end{enumerate}
If they start with the state $\ket{\Psi}^{AB}=\sqrt{c}\ket{\Phi_p} + \sqrt{1-c} \ket{\Phi_p^{\perp}}$, after steps $1$-$3$, the state is
\begin{align}
    \left(2\ketbra{s}{s}-\iden\right)\cdot \frac{1}{\sqrt{d}} \sum_{j=1}^d \ket{j}\ U_j \ot U^*_j \ket{\Psi} &= \frac{2}{d}\ket{s}\sum_{j=1}^d U_j \ot U^*_j \ket{\Psi}-\frac{1}{\sqrt{d}} \sum_{j=1}^d \ket{j}\ U_j \ot U^*_j \ket{\Psi}\nn\\
    &\approx_{\e} 2 \sqrt{c} \ket{s}\ket{\Phi_p}-\frac{1}{\sqrt{d}} \sum_{j=1}^d \ket{j}\ U_j \ot U^*_j \ket{\Psi},
\end{align}
where we used \eqref{eq:t8} to get to the last line. By the end of step $5$, we have 
\begin{align}
    V^* \ot V \left(2 \sqrt{c} \ket{s}\ket{\Phi_p}-\frac{1}{\sqrt{d}} \sum_{j=1}^d \ket{j}\ U_j \ot U^*_j \ket{\Psi} \right)&= 2\sqrt{\frac{c}{d}}\sum_{j=1}^d \ket{j}\ U^*_j \otimes U_j \ket{\Phi_p}-\ket{s} \ket{\Psi}\nn\\
    &=\ket{s}\left(2\sqrt{c} \ket{\Phi_p}-\ket{\Psi}\right)\nn\\
    &=\ket{s}\left(\sqrt{c}\ket{\Phi_p}-\sqrt{1-c} \ket{\Phi_p^{\perp}}\right)\nn\\
    &=\ket{s} \ot (2\Phi_p-\iden) \ket{\Psi}.
\end{align}
We used the fact that $U_j^* \ot U_j \ket{\Phi_p}=\iden \ot U_j^{\dagger} U_j \ket{\Phi_p}=\ket{\Phi_p}$ to get the second line. 
\end{proof}
\begin{thm}[Communication protocol for Hamiltonian simulation]\label{thm:Communication protocol for Hamiltonian simulation} There exists a communication protocol between Alice and Bob, summarized in \protoref{H_sim protocol}, that uses a shared maximally entangled state $\ket{\Phi_p}$ for an arbitrarily large $p$ and $O(t|\partial A|\log (t|\partial A|/\e))$ extra qubits of communication and implements an operator $W_t$ such that  
\begin{align}
\Norm{W_t\ket{\Phi_p} \ket{\psi}-\ket{\Phi_p} e^{-it(H_A+H_{\partial A}+H_B)}\ket{\psi}}\leq \e.
\end{align}
\end{thm}
\begin{proof}[Proof of \thmref{Communication protocol for Hamiltonian simulation}]
The evolution is divided into $L$ segments of length $\tau$ as in \eqref{eq:t2}. In each segment, the operator $e^{-i\tau(H_A+H_{B})}=e^{-i\tau H_A} e^{-i\tau H_B}$ can be implemented without any communications. Thus we focus on the cost of performing the time-ordered propagator $\cT\left[\exp\left({-i\int_0^{\tau} H_I(s)ds}\right)\right]$. 

Following \cite{low2018Hamiltonian,berry2015simulating}, we use the LCU method to simulate the time-ordered operator given as a sum of unitaries in \eqref{eq:t6}. In the original LCU algorithm, to implement a sum of unitaries such as $\sum_{\ell \in I_{M,K}}\a_{\ell}\ v^{(A)}_{\ell} \ot v^{(B)}_{\ell}$. Alice and Bob need to share (and later reflect about) the state $\sum_{\ell\in I_{M,K}} \sqrt{\a_{\ell}} \ket{\ell} \ket{\ell}$. In general sharing such a state results in extra entanglement spread between the parties. To avoid this, we modify the sum in \eqref{eq:t6} so that all $\a_{\ell}$ are equal and Alice and Bob can instead use their shared maximally entangled state $\ket{\Phi_p}$ which has zero entanglement spread. 

We achieve this by rounding off the coefficients $\a_{\ell}$ to the nearest multiple of $k_{\d}=2^{-\lceil \log(\d^{-1}) \rceil}$ denoted by $\tilde{\a}_{\ell}$ such that $|\a_{\ell}-\tilde{\a}_{\ell}|\leq \d \ll 1$. The choice of $\d$ depends on $M$ and $K$. In particular as $M,K\rightarrow \infty$, we have $\d\rightarrow 0$. We can re-express the sum in \eqref{eq:t6} by repeating each term $v^{(A)}_{\ell} \ot v^{(B)}_{\ell}$ for $\tilde{\a}_{\ell}/k_{\d}$ times. This means for a fixed $M$, $K$, and $\d$, Alice and Bob wish to implement the sum 
\begin{align}
 k_{\d} \cdot \sum_{\ell=1}^{p}\ v^{(A)}_{\ell} \ot v^{(B)}_{\ell}
\end{align}
with some extended set of indices $\ell$ with size $p\leq \sum_{\ell \in I_{M,K}}\lceil \a_{\ell} /k_{\d}\rceil<e^{\tau \sum_{j=1}^J \b_j}/k_{\d}$ (using \eqref{eq:normupbound}). The simulation protocol consists of the following steps summarized in \protoref{H_sim protocol}:
\begin{enumerate}
    \item Alice and Bob perform the following operator on the state $\ket{\psi}$ and the maximally entangled state $\ket{\Phi_p}=\frac{1}{\sqrt{p}}\sum_{\ell=1}^{p} \ket{\ell}\ket{\ell}$ shared between them:
\begin{align}
  \SEL & =\SEL_A \otimes \SEL_B\nn
  \\
  \SEL_A & =\sum_{\ell=1}^{p} \ketbra{\ell}{\ell}^{(A)} \ot v^{(A)}_{\ell} \nn
  \\
  \SEL_B & = \sum_{\ell=1}^{p} \ketbra{\ell}{\ell}^{(B)} \ot v^{(B)}_{\ell}.\nn
\end{align}
To implement this, Alice (Bob) applies the unitary $v_{\ell}^{(A)}$ ($v_{\ell}^{(B)}$) on their spins conditioned on register $\ket{\ell}^{(A)}$ ($\ket{\ell}^{(B)}$). Hence, the operator $\SEL$ can be implement by the parties only using local unitaries. Their state after this step is
\begin{align}
  \SEL \ket{\Phi_p}\ket{\psi}= \frac{1}{\sqrt p}
  %\frac{1}{(\sum_{\ell \in I_{M,K}}\a_{\ell})^{1/2}}
  \sum_{\ell=1}^{p} \sqrt{k_{\d}}\ \ket{\ell}^{(A)}\ket{\ell}^{(B)} \ot \left(v^{(A)}_{\ell}\ot v^{(B)}_{\ell} \right) \ket{\psi}\label{eq:t7}
\end{align}
     \item  Next, the parties implement the oblivious amplitude amplification to \emph{turn} the state $\SEL \ket{\Phi_p}\ket{\psi}$ into the desired state $\ket{\Phi_p} \left(\sum_{\ell} k_{\d}\ v^{(A)}_{\ell} \ot v^{(B)}_{\ell}\right) \ket{\psi}$. This means they apply the rotation operator $-\SEL (2\Phi_p-\iden) \SEL^{\dagger} (2\Phi_p-\iden)$. It is shown in \cite{berry2015simulating} that if $\sum_{\ell\in I_{M,K}} \a_{\ell}=2$, one application of this operator suffices. Using an extra ancillary qubit and setting the number of segments $L= t\sum_j \b_j/\ln(2)=O\left(t|\partial A|\right)$, we can always assume $\sum_{\ell \in I_{M,K}}\tilde\a_{\ell}=2$.   
    
    Similar to step 1, the operators $-\SEL$ and $\SEL^{\dagger}$ are performed locally. The reflection operator $2\Phi_p-\iden$ is performed using the protocol in \thmref{protocol for reflection around the MES}. For an overall error of $\e$, we need error $\e/L$ per segment, which requires $O\left(\log(L/\e)\right)=O\left(\log(t|\partial A|/\e)\right)$ qubits of communication per segment.
    
    In the described protocol, we can take $M,K \rightarrow \infty$. By doing so  $\d\rightarrow 0$. This will only increase the size $D$ of the shared maximally entangled state $\ket{\Psi_D}$ and  not the qubits communicated when implementing $2\Phi_p-\iden$.
    
    Their state after performing amplitude amplification is
\begin{align}
  \left(-\SEL (2\Phi_p-\iden) \SEL^{\dagger} (2\Phi_p-\iden)\right) \SEL  \ket{\Phi_p} \ket{\psi}
\approx_{\e/L}
  \ket{\Phi_p} \left(\sum_{\ell'} k_{\d}\ v^{(A)}_{\ell'} \ot v^{(B)}_{\ell'}\right) \ket{\psi}.
\end{align}

    \item  Alice performs $e^{-i\tau H_A}$ and Bob performs $e^{-i\tau H_B}$ on their spins. 
    \item They repeat the steps 1-3 for $L=O\left(t|\partial A|\right)$ times.
\end{enumerate}
 Hence, the number of qubits exchanged during the whole protocol is $O(t|\partial A|\log (t|\partial A|/\e))$.
\end{proof}
\begin{protocol}[H]
\caption{Hamiltonian simulation protocol between Alice and Bob}
\label{proto:H_sim protocol}
\textit{Input:} Unbounded shared maximally entangled state $\ket{\Phi}$, a shared state $\ket{\psi}$, $H$ and $t$.  \\
\textit{Goal:} Implement $W_t$ such that  $\norm{W_t \ket{\Phi} \ket{\psi}-\ket{\Phi}e^{-itH}\ket{\psi}}\leq \e$. \\
\textit{Procedure:}\\
For $L=O\left(t|\partial A|\right)$ times, perform the following protocol:
\begin{enumerate}
  \item Alice and Bob implement $\SEL \ket{\Phi}\ket{\psi}$ by applying local controlled-unitaries, 
  \item  Using the protocol in \thmref{protocol for reflection around the MES}, Alice and Bob approximately perform the rotation operator $\left(-\SEL (2\Phi-\iden) \SEL^{\dagger} (2\Phi-\iden)\right)$ using $O\left(\log(t|\partial A|/\e)\right)$ qubits of communication. 
  \item Alice applies $e^{-i H_At/L}$ and Bob applies $e^{-i H_Bt/L}$ locally.
    \end{enumerate}
\end{protocol}

 \subsection{A communication protocol for measuring the ground state}
In this section, we use the Hamiltonian simulation protocol of \secref{Interaction picture protocol} along with the quantum phase estimation algorithm to approximately implement the two-outcome measurement $\{\Omega,\iden-\Omega\}$.

The phase estimation algorithm is an operator $\PHASE$ that uses $f=O\left(\log(1/\g)\right)$ ancillary registers and queries an oracle that implements $e^{-itH}$ for $t$ range from $0$ to $O(1/\g)$.  The action of this operator on the eigenstates of $H$ is (assuming thah the $e^{-itH}$ oracle is perfect):
\begin{align}
    \PHASE\ \ket{0} \ket{\Omega} &= \ket{0}\ket{\Omega} \nn\\
    \PHASE\ \ket{0} \ket{E_i}&= \sqrt{\s_i} \ket{0}\ket{E_i} +\sqrt{1-\s_i} \ket{0^{\perp}}\ket{E_i},\quad \forall i\geq 1,\nn
\end{align}
where the $|\s_i|$ are all less than some univeral constant $\sigma<1$. Be repeating this operator for $O\left(\log(1/\D)\right)$ times, we can reduce the error to $\leq \D$. We define the two-outcome POVM \begin{align}
    K=\PHASE^{\dagger} (\ketbra{0}{0}\ot \iden)\PHASE.\nn
\end{align}
We see that \cite[Equation 10-12]{Somma_Phase_estimation}
\begin{align}
\norm{\left(K-\iden\otimes\ketbra{\Omega}{\Omega}\right)\ket{0}\ket{\psi}}\leq \D.\label{eq:s11}
\end{align}

In our EPR-assisted communication protocol, the Hamiltonian simulation operator is implemented approximately, hence introducing an additional error in \eqref{eq:s11}. In the following, we give the details of this protocol and its analysis which shows how to implement the two-outcome measurement $\{K, \iden-K\}$.
\begin{protocol}[H]
\caption{Protocol for measuring the ground state}
\label{proto:H_gs_reflection}
\textit{Input:} Unbounded shared maximally entangled state $\ket{\Phi}$, ancillary state $\ket{0}$, a shared state $\ket{\psi}$, and the Hamiltonian $H$.  \\
\textit{Goal:} Perform POVM $\{K,\iden-K\}$ such that  $\norm{\left(K-\iden\otimes\ketbra{\Omega}{\Omega}\right)\ket{0}\ket{\Phi}\ket{\psi}}\leq \D$. \\
\textit{Procedure:}
\begin{enumerate}
    \item For $k=O\left(\log(1/\D)\right)$ times, repeat the following steps i.-iv. to perform the operator $\PHASE$ in parallel:
\begin{itemize}
  \item[i.] Alice prepares the state $\frac{1}{2^{f/2}}\sum_{j=0}^{2^f-1}\ket{j}_{a_k}\ket{j}_{b_k}$ with $f=O(\log 1/\gamma)$ and shares register $b_k$ with Bob. 
  \item[ii.] Conditioned on registers $\ket{j}_{a_k}\ket{j}_{b_k}$ Alice and Bob implement the Hamiltonian simulation protocol $W_j$ (\protoref{H_sim protocol}).
  \item[iii.] Bob returns his register $b_k$ to Alice.
  \item[iv.] Alice uncomputes register $b_k$ and performs the Fourier transform $F^{\dagger}$ on her register $a_k$.
    \end{itemize}
\item Alice applies the measurement $\left\{\ketbra{0}{0}, \iden-\ketbra{0}{0}\right\}$ on her registers $\Ot_{i=1}^k a_i$.
\item Similar to step $1$, Alice and Bob perform $\PHASE^{\dagger}$.
\end{enumerate}
\end{protocol}

\begin{thm}[Communication protocol for measuring the ground state]
\label{thm:reflectaroundgs}
Alice and Bob can implement a measurement $\left\{K, \iden-K\right\}$ such that $$\norm{\left(K-\iden\otimes\ketbra{\Omega}{\Omega}\right)\ket{0}\ket{\Phi}\ket{\psi}}\leq \D$$ while sharing unlimited EPR pairs $\ket{\Phi}$ and with the communication cost 
\begin{align}
O\left(\frac{|\partial A|}{\g} \log\left(\frac{|\partial A|}{\g \D}
\log \frac{1}{\D}\right) \log \frac{1}{\D}\right).\label{eq:t10}
\end{align}
\end{thm}
\begin{proof}
In each application of $\PHASE$, Alice prepares and shares $f=O\left(\log 1/\g\right)$ ancillary registers with Bob (step 1 in \protoref{H_gs_reflection}). (It is possible that EPR testing could be used to save communication in steps i, iii, iv but we do not investigate this since the communication cost is dominated by step ii.) The phase estimation algorithm is repeated $O\left(\log(1/\D)\right)$ times. In each application, the Hamiltonian simulation \protoref{H_sim protocol} is run once. Hence, there are a total of $O\left(\log(1/\D)\right)$ calls to this protocol. For an overall error of $O(\D)$, the error in performing the Hamiltonian simulation \protoref{H_sim protocol} is set as 
\begin{align}
\e=O\left(\frac{\D}{\log(1/\D)}\right).\nn
\end{align}
According to \thmref{Communication protocol for Hamiltonian simulation}, the communication cost of implementing \protoref{H_sim protocol} is bounded by $O\left(\frac{|\partial A|}{\g} \log\left(\frac{|\partial A|}{\g \e}\right)\right)$. Adding these costs, we get \eqref{eq:t10}. This protocol achieves the desired measurement. 
\end{proof}
\section{AGSP for lattice Hamiltonians using Chebyshev polynomials}\label{sec:AGSP for lattices}

The idea of truncation, introduced in \cite{area_law_subexponential_algo}, allows one to control the norm of Hamiltonian away from a bipartite cut. In this section we (1) review the previous techniques for truncation in frustration-free and general Hamiltonians and (2) adapt them from 1D systems to an arbitrary lattice. First, we explain how to perform truncation in the frustration free case.

\subsection{Truncation: frustration-free case}
 
Without loss of generality, we assume that $h_k$ are projectors and the ground energy $E_0=0$. Let $\DL$ be the detectability lemma operator \cite{detectability_Aharonov} corresponding to $H$ defined as follows.
\begin{definition}[Detectability lemma operator]
Consider a partition of the terms of the Hamiltonian $H=\sum_{i=k}^N h_k$ into $w$ groups $\{T_1, \ldots T_w\}$, where the terms in each group mutually commute. For a finite dimensional lattice, $w$ is a constant. The detectability lemma operator is defined by
\begin{align}
\DL= \prod_{\alpha=1}^w\left(\prod_{k\in T_\alpha}(\iden-h_k)\right). \nn
\end{align}
\end{definition}
The operator $\DL$ defines an AGSP for the Hamiltonian $H$. In particular, since $H$ is frustration-free, the terms $\iden-h_k$ preserve the ground state and we have $\DL\ket{\Omega}=\ket{\Omega}$. Since $\norm{\DL} \leq 1$, when this operator is applied to the states orthogonal to $\ket{\Omega}$, their norm shrinks by a factor $\leq 1$. More precisely, it is shown in \cite{detectability_Aharonov,Simple_proof_detectabilit} that we have
\begin{equation}
\label{DLlemma}
\DL\ket{\Omega}=\ket{\Omega}, \quad \|\DL-\ketbra{\Omega}{\Omega}\|\leq \frac{1}{1+\g/g^2}.
\end{equation}
where $\g$ is the spectral gap and $g$ is the number of interactions in the Hamiltonian not commuting with a given interaction $h_k$. 

\begin{rem}
For a $D$-dimensional lattice with a $\k$-local Hamiltonian, we have $w\leq (2D)^{2\k}$ and $g\leq \k(2D)^{\k-1}$ (see for instance,  \cite[Section II]{Simple_proof_detectabilit})
\end{rem}

In order to obtain a truncation for the Hamiltonian $H$, we consider a slightly different AGSP than $\DL$. Consider any bipartition $(A: B)$ of the lattice. Let $\Pi_{A}$ be the projector onto the ground space of the Hamiltonian $H_A$ and $\Pi_{B}$ be the projector onto the ground space of $H_{B}$. Let $C$ be the set of all interactions contained within $\partial_w A$. Using the ``absorption'' argument from \cite{detectability_Aharonov,one-dimensional_area_law}, the following equality can be shown:
\begin{equation}
\label{newDL}
(\Pi_{A}\Pi_{B})\DL = \Pi_{A}\Pi_{B}\prod_{\alpha=1}^w\left(\prod_{k\in T_\alpha\cap C}(\iden-h_k)\right),
\end{equation}
which can be verified by noticing that $(\iden-h_k)\Pi_R=\Pi_R$, where $\Pi_R$ is the projector onto the ground space of a region $R$ on which $h_k$ is supported, we can absorb terms from $\DL$ into $\Pi_{A}\Pi_{B}$ except for those that are hindered due to the boundary.  

By applying \eqref{newDL} in \eqref{DLlemma}, and using $\Pi_{A}\Pi_{B}\ket{\Omega}=\ket{\Omega}$, we conclude that $(\Pi_{A}\Pi_{B})\DL$ is also an AGSP, i.e.
\begin{equation}
\label{newDL2}
\NORM{\Pi_{A}\Pi_{B}\prod_{\alpha=1}^w\left(\prod_{k\in T_\alpha\cap C}(\iden-h_k)\right)-\ketbra{\Omega}{\Omega}} \leq \frac{1}{1+\g/g^2}.
\end{equation}
Next, we use this operator to truncate the Hamiltonian $H$ outside some region $A$.
\begin{thm}[Truncation in the frustration free case]
The truncation of a frustration Hamiltonian $H$ with respect to the partition $(A:B)$ is defined by:
\begin{equation}
\label{eq:FFtrunc}
\tilde{H}= \sum_{\alpha, k\in T_\alpha \cap C} h_k  + (\iden-\Pi_{A}) + (\iden-\Pi_{B}),
\end{equation}
where It holds that (i) $\tilde{H}$ is frustration free, (ii) $\norm{\tilde{H}}\leq 2+ w^2|\partial A|$, (iii) the spectral gap of $\tilde{H}$ is $\leq \g/4g^2$.
\end{thm}
\begin{proof}
One can see that $\ket{\Omega}$ is a ground state of $\tilde{H}$. In order to lower bound the spectral gap of $\tilde{H}$, we use the fact that the detectability lemma operator and the Hamiltonian have very similar spectral gaps. This was described as a converse to the detectability lemma in \cite{Simple_proof_detectabilit}. More precisely, using Theorem 1.1b of \cite{Quantum_union_bounds_Gao}, we obtain that for any state $\ket{\psi}$,
\begin{align}
4\bra{\psi}\tilde{H}\ket{\psi} \geq 1-\NORM{\Pi_{A}\Pi_{B}\prod_{\alpha=1}^w\left(\prod_{k\in T_\alpha\cap C}(\iden-h_k)\right)\ket{\psi}}^2.\nn
\end{align}
If $\ket{\psi}$ is orthogonal to $\ket{\Omega}$, \eqref{newDL2} ensures that 
$$4\bra{\psi}\tilde{H}\ket{\psi} \geq 1-\left(\frac{1}{1+\g/g^2}\right)^2 \geq \frac{\g}{g^2}.$$
Thus the spectral gap of $\tilde{H}$ is at least $\frac{\g}{4g^2}$. Lastly, we have  has norm at most 
\begin{align}
    \norm{\tilde{H}} \leq \norm{\iden -\Pi_A}+\norm{\iden-\Pi_B}+ \norm{\sum_{\alpha, k\in T_\alpha \cap C} h_k} \leq 2+w^2|\partial A|\nn
\end{align}
\end{proof}
\subsection{Truncation: frustrated case}

Here, we consider truncation in the more general case of frustrated Hamiltonians. This is first achieved in \cite{area_law_subexponential_algo}. We will directly use the following theorem from \cite{Kuwahara_area_law19}, which built upon \cite{Arad_connecting_global_local_dist}. For a partition $(A:B)$, the truncation in  \cite{Arad_connecting_global_local_dist} is defined by removing the high energy spectrum of $H_B$. The improvement in \cite{Kuwahara_area_law19} allows one to truncate both $H_A$ and $H_B$ while leaving the boundary term $H_{\partial_w A}$ untouched, where $w=O(1)$.

\begin{definition}[Truncation of $H$ up to energy $\xi$]
Fix a bipartition $(A:B)$ such that $H=H_A+H_B+H_{\partial_w A}$ and $w=O(1)$. Let $\Pi_{A}^{< \xi}$ and  $\Pi_A^{\geq \xi}$ denote the projectors onto the eigenstates of $H_A$ with energy $< \xi$ and $\geq \xi$ respectively. Similarly, we assign $\Pi_B^{<\xi}$ and $\Pi_B^{\geq \xi}$ for region $B$. The truncation of $H_A$ (or $H_B$) up to energy $\xi$ is defined as
\begin{align}
    \tilde{H}_A=H_A\Pi_A^{<\xi}+ \xi \Pi_A^{\geq \xi}.\nn
\end{align}
Moreover, the truncation of $H$ up to energy $\xi$ with respect to the partition $(A:B)$ is defined by
\begin{align}
    \tilde{H}=\tilde{H}_A+\tilde{H}_B+H_{\partial_w A}.\nn
\end{align}
\end{definition}

\begin{thm}[Truncation in the frustrated case, cf. \cite{Kuwahara_area_law19}, Theorem 5]
\label{frustrunc}
Let $\tilde{H}$ be the truncation of the Hamiltonian $H$ with respect to the partition $(A:B)$ up to energy $\xi=99\left(w^2|\partial_w A| + \log\frac{1}{\varepsilon\g}\right)$, where $\varepsilon > e^{-0.1|\partial_w A|}$ and $w=O(1)$. Then, it holds that
\begin{itemize}
    \item[i.] The spectral gap of $\tilde{H}$ is at least $\frac{\g}{2}$,
    \item[ii.] Ground state $\ket{\Omega'}$ of $\tilde{H}$ satisfies $|\braket{\Omega}{\Omega'}|\geq 1-\varepsilon$,
    \item[iii.] $\|\tilde{H}-E'_0\iden\| \leq 100w^2|\partial_w A|$, where $E'_0$ be the ground energy of $\tilde{H}$.
\end{itemize}
\end{thm}

\subsection{Chebyshev-AGSP}

Previous subsections show that the Hamiltonian $H$ can be truncated to $\tilde{H}$ such that the spectral gap stays $\geq \g/2$ and norm of the Hamiltonian is at most $E'_0+o(w^2|\partial_w A|)$. Furthermore, the ground state $\ket{\Omega'}$ is close to $\ket{\Omega}$ with fidelity at least $e^{-0.1|\partial_w A|}$. This means we can instead construct an AGSP for the ground state $\ket{\Omega'}$ of the truncated Hamiltonian $\tilde{H}$.

\begin{definition}[Chebyshev-AGSP]
 Define the Chebyshev-AGSP as the following polynomial of $\tilde{H}$ of degree $q$: $$Q_q(\tilde{H})= \frac{1}{T_{q}\left(1+\frac{2E'_1-2E'_0}{E'_{\max}-E'_1}\right)}T_q\left(1+\frac{2E'_1-2\tilde{H}}{E'_{\max}-E'_1}\right),$$ where $E'_{\max}$ is the largest eigenvalue of $\tilde{H}$ and $T_q$ is the degree $q$ Chebyshev polynomial of first kind defined by $T_q(\cos \theta)=\cos(q\theta)$. 
\end{definition}
\begin{thm}\label{thm:Chebyshev-AGSP}
There is a constant $\xi$ depending on the geometry of the Hamiltonian such that if we let $q=\sqrt{w^2g^2\frac{|\partial_w A|}{\xi^2\gamma}}\log\frac{4}{\Delta}$, then $Q_q(\tilde{H})$ is a $(D,\D)$-AGSP with respect to the partition $(A:B)$ (see \defref{AGSP}). That is, $\|\ketbra{\Omega}{\Omega}-Q_q(\tilde{H})\| \leq \Delta$ and the Schmidt rank is bounded by 
\begin{equation}
\label{eq:schmidtrank}
D=\exp\left(\sqrt{\frac{|\partial_w A|}{\gamma}}\cdot \log\frac{4}{\Delta} \cdot \frac{wg}{c}\cdot\log\left(\frac{|\partial_w A|^2}{\g}w^2g^2s^b\log^2(\frac{4}{\Delta})\frac{1}{\xi^2}\right)\right),
\end{equation}
\end{thm}
Before stating the proof of \thmref{Chebyshev-AGSP}, we need the following lemma:
\begin{lem}[Adapted from \cite{area_law_subexponential_algo}]
\label{lem:SRbound}
The Schmidt rank of $Q_q(\tilde{H})$ is at most $D\leq e^{q\log\left(q^2s^b|\partial_w A|\right)}$.
\end{lem}
\begin{proof}
The Schmidt rank of $Q_q(\tilde{H})$ is at most $q$ times the Schmidt rank of $(\tilde{H})^q$. Thus, we upper bound the latter. In both the frustration-free (\eqref{eq:FFtrunc}) and the frustrated case (Theorem \ref{frustrunc}), we can write $\tilde{H}= H_{\partial_w A} + X + Y$, where $X$ is an operator supported on region $A$ and $Y$ is an operator acting on region $B$. Consider the following expansion:
$$(\tilde{H})^q = \sum_{\ell=1}^{q+1} \left(X^{f_1}Y^{g_1}\right)H_{\partial_w A}\left(X^{f_2}Y^{g_2}\right)H_{\partial_w A}\left(X^{f_3}Y^{g_3}\right)\ldots H_{\partial_w A}\left(X^{f_\ell}Y^{g_\ell}\right).$$
In each term, $H_{\partial_w A}$ occurs $\ell-1$ times, and the tuple of non-negative integers $\left(f_1, g_1, \ldots f_\ell, g_\ell\right)$ satisfies
$$\sum_{i=1}^{\ell}(f_i+g_i) = q-\ell+1.$$
The number of possible such tuples is equal to ${q+\ell+1 \choose 2\ell} \leq {2q+2 \choose 2\ell} \leq (q+1)^{\ell}$.  Since none of $X^{f_i}Y^{g_i}$ change the Schmidt rank across the bipartition, and $H_{\partial_w A}$ changes the Schmidt rank by at most $s^{b}|\partial_w A|$, we obtain that the Schmidt rank of $(\tilde{H})^q$ is at most
$$(s^{b}|\partial_w A|)^{q}\cdot (q+1)^{q+2}\leq e^{q\log(q^2s^b|\partial_w A|)}.$$
This completes the proof.
\end{proof}
\begin{proof}[Proof of \thmref{Chebyshev-AGSP}]
By a result of \cite{area_law_subexponential_algo}, we have $$\|Q_q(\tilde{H})-\ketbra{\Omega'}{\Omega'}\| \leq 2 e^{-2q\sqrt{\frac{E'_1-E'_0}{E'_{\max}-E'_1}}} \leq 2 e^{-\xi q\sqrt{\frac{\g}{w^2g^2|\partial_w A|}}},$$ where $\xi$ is an absolute constant determined by the lattice structure and the locality of the Hamiltonian. For a given $\Delta$, choose $q=\sqrt{w^2g^2\frac{|\partial_w A|}{\xi^2\gamma}}\log\frac{4}{\Delta}$. Then 
\begin{equation}
\label{eq:reflectapprox}
\|\ketbra{\Omega}{\Omega}-Q_q(\tilde{H})\| \leq \Delta.\nn
\end{equation}

By plugging in the choice of $q=\sqrt{w^2g^2\frac{|\partial_w A|}{\xi^2\gamma}}\log\frac{4}{\Delta}$ in \lemref{SRbound}, we see that the Schmidt rank is at most 
\begin{equation}
% \label{eq:schmidtrank2}
D= \exp\left(\sqrt{\frac{|\partial_w A|}{\gamma}}\cdot \log\frac{4}{\Delta} \cdot \frac{wg}{\xi}\cdot\log\left(\frac{|\partial_w A|^2}{\g}w^2g^2s^b\log^2(\frac{4}{\Delta})\frac{1}{\xi^2}\right)\right).\nn
\end{equation}
\end{proof}

\appendix
\section{Proof of \propref{compressionlemma}}\label{sec:compression} 
Suppose $\alpha_i=\frac{w_i}{N}$ are rational numbers, with $w_i$ and $N$ nonnegative integers. This can be assumed with arbitrarily small error. We re-write
$$K= \frac{1}{N}\sum_i\sum_{j=1}^{w_i} U_i \otimes V_i.$$
Let $M=\sum_i w_i$. Introduce a maximally entangled state $$\ket{\Phi}=\frac{1}{\sqrt{M}}\sum_{i}\sum_{j=1}^{w_i} \ket{i,j}_a\ket{i,j}_b,$$ 
where Alice's and Bob's auxiliary registers are denoted by $a$ and $b$ respectively. This leads to the following representation of $K$:
$$K\otimes \ketbra{\Phi}{\Phi}_{ab}= \frac{M}{N}\left(\ketbra{\Phi}{\Phi}_{ab}\left(\sum_{i,j}\ketbra{i,j}{i,j}_a\otimes U_i\right)\otimes\left(\sum_{i,j}\ketbra{i,j}{i,j}_b\otimes V_i\right)\ketbra{\Phi}{\Phi}_{ab}\right).$$
From Theorem \ref{thm:protocol for reflection around the MES}, there exists an operator $L$ with Schmidt rank $\frac{1}{\epsilon^{\mathcal{O}(1)}}$, such that 
$$\|L-\ketbra{\Phi}{\Phi}_{ab}\|\leq \epsilon.$$
Letting $\epsilon=\frac{N\Delta}{M}$, we obtain the following approximation to $K\otimes \ketbra{\Phi}{\Phi}_{ab}$:
$$K'=\frac{M}{N}\left(L\left(\sum_{i,j}\ketbra{i,j}{i,j}_a\otimes U_i\right)\otimes\left(\sum_{i,j}\ketbra{i,j}{i,j}_b\otimes V_i\right)\ketbra{\Phi}{\Phi}_{ab}\right)$$
such that
$$\norm{K'-K\otimes \ketbra{\Phi}{\Phi}_{ab}}\leq \Delta \implies \|K'\left(\iden\otimes \ket{\Phi}_{ab}\right)-\ketbra{\Omega}{\Omega}_{AB}\otimes \ket{\Phi}_{ab}\|\leq 2\Delta.$$
When $K'$ acts on a state $\ket{\psi}_A\ket{\psi}_B\ket{\Phi}_{ab}$, the Schmidt rank is increased by at most the Schmidt rank of $L$, which is 
$$\left(\frac{M}{N\Delta}\right)^{\mathcal{O}(1)} = \left(\frac{1}{\Delta}\sum_i\frac{w_i}{N}\right)^{\mathcal{O}(1)} =  \left(\frac{\sum_i \alpha_i}{\Delta}\right)^{\mathcal{O}(1)}.$$
By definition, this is the Schmidt rank of the EPR-assisted AGSP $K'\left(\iden\otimes \ket{\Phi}_{ab}\right)$. This completes the proof. 

 \section*{Acknowledgements}
AWH thanks Dorit Aharonov for insightful discussions regarding this project and raising the question of the connection between communication complexity of measuring the ground state and ground state entanglement. AA thanks David Gosset for discussions on applications of quantum algorithms to area laws. MS thanks Zeph Landau, Anand Natarajan, and Umesh Vazirani for helpful discussions. AA is supported by the Canadian Institute for Advanced Research, through funding provided to the Institute for Quantum Computing by the Government of Canada and the Province of Ontario. Perimeter Institute is also supported in part by the Government of Canada and the Province of Ontario.  AWH was funded by NSF grants CCF-1452616, CCF-1729369, PHY-1818914,
ARO contract W911NF-17-1-0433 and a Samsung Advanced
Institute of Technology Global Research Partnership.  MS was funded by NSF grant CCF-1729369.

\bibliographystyle{alpha}
\bibliography{main}

\end{document}